\documentclass[amsmath,amssymb,showkeys, showpacs, superscriptaddress]{revtex4-1}
\usepackage{graphicx}
\usepackage{mathrsfs}
\usepackage{relsize}
\DeclareMathAlphabet{\mathpzc}{OT1}{pzc}{m}{it}
\usepackage{latexsym}
\usepackage{amsmath,amsfonts,MnSymbol}
\usepackage{amsthm}
\usepackage{rotating}

\usepackage{amsmath,tikz}

\newtheorem{theorem}{Theorem}[section]

\usepackage{color} %FOR COLOURED TEXT {\color{<arg>} <text> }
\usepackage{cancel}

\begin{document}

% ---------------------------
%\hspace*{4 in}CUQM-143\\
%\hspace*{4 in} polya [9$^{\rm th}$ Jan 2012]
% 4 Jan 2012 version from Nasser
\vspace*{0.4 in}
% ---------------------------
\title{The $d$-dimensional softcore Coulomb potential and the\smallskip \\generalized confluent Heun equation}
\author{Richard L. Hall}
\email{richard.hall@concordia.ca}
\affiliation{Department of Mathematics and Statistics, Concordia University,
1455 de Maisonneuve Boulevard West, Montr\'{e}al,
Qu\'{e}bec, Canada H3G 1M8}

\author{Nasser Saad}
\email{nsaad@upei.ca}
\affiliation{School of Mathematical and Computational Sciences,
University of Prince Edward Island, 550 University Avenue,
Charlottetown, PEI, Canada C1A 4P3.}
\author{ Kyle R. Bryenton }
\email{kbryenton@upei.ca}
\affiliation{School of Mathematical and Computational Sciences,
University of Prince Edward Island, 550 University Avenue,
Charlottetown, PEI, Canada C1A 4P3.}
\begin{abstract}
\noindent \textbf{Abstract:} 
An analysis of the generalized confluent Heun equation $(\alpha_2r^2+\alpha_1r)\,y''+(\beta_2r^2+\beta_1r+\beta_0)\,y'-(\varepsilon_1r+\varepsilon_0)\,y=0$ in $d$-dimensional space, where  $\{\alpha_i, \beta_i, \varepsilon_i\}$ are real parameters, is presented. With the aid of these general results, the quasi exact solvability of the Schr\"odinger eigen problem generated  by the softcore Coulomb potential $V(r)=-e^2Z/(r+b),\, b>0$, is explicitly resolved.   Necessary and sufficient conditions for polynomial solvability are given.  A three-term recurrence relation is provided to generate the coefficients of polynomial solutions explicitly.  We prove that these polynomial solutions are sources of finite sequences of orthogonal polynomials. Properties such as recurrence relations, Christoffel-Darboux formulas, and the moments of the weight function are discussed. We also reveal a factorization property of these polynomials which permits the construction of other interesting related sequences of orthogonal polynomials.  

\end{abstract}



\maketitle
%%%%%%%%%%%%%%%%%%%%%%%%%%%%%%%%%%%%%%
\section{Introduction} 
%%%%%%%%%%%%%%%%%%%%%%%%

\noindent We study the $d$-dimensional radial Schr\"odinger equation
(in atomic units $\hbar=2\mu=1$),
\begin{equation}\label{eq1}
\left[-\dfrac{d^2}{dr^2}+\dfrac{(k-1)(k-3)}{4r^2}+V(r)-E\right] \psi(r)=0 \, , \quad
\int\limits_0^\infty \left\{\psi(r)\right\}^2dr=1 \, , \quad
\psi(0)=0 \, , \quad 
V(r)=-\dfrac{e^2\,Z}{r+b},
\end{equation}
where $d > 1$, $\ell$ is the angular-momentum quantum number, $ k = 2\,\ell +d,$  and $b >0$ is a cut-off parameter.  Thus $E = E(e^2 Z, b)$. The number of potential parameters can be reduced by  scaling: for example the change of variable $r \rightarrow b\,r$  with $v = b\, e^2\, Z$ and  $E_{{\mathfrak n}k}(v) = b^2 E(e^2 Z, b)$ transforms equation \eqref{eq1} into the following
\begin{equation}\label{eq2}
\left[-\dfrac{d^2}{dr^2}+ \dfrac{(k-1)(k-3)}{4r^2} - \dfrac{v}{r+1} \right] \psi_{{\mathfrak n}k}(r)=E_{{\mathfrak n}k}(v)\, \psi_{{\mathfrak n}k}(r) ,\quad v > 0,
\end{equation}
where ${\mathfrak n} = 0,1,2,\dots$ is the number of nodes in the radial wave function $\psi_{{\mathfrak n}k}(r)$.
\vskip0.1true in
\noindent  The relation between Schr\"odinger's equation  \eqref{eq2} for $\ell\neq 0$ and the confluent Heun equation was pointed earlier by H. Exton \cite{exton}, (see also \cite{saad2009}). In the present work we study this connection in more detail. This interrelationship in turn involves a deeper analysis of the generalized confluent Heun equation, 
\begin{align}\label{eq3}
r\,(\alpha_2r+\alpha_1)\,f''(r)+(\beta_2r^2+\beta_1r+\beta_0)\,f'(r)-(\varepsilon_1r+\varepsilon_0)\,f(r)=0,
\end{align}
 that can be rewritten in the self-adjoint form
\begin{align*}
\dfrac{1}{r^{\frac{\beta_0}{\alpha_1}-1} e^{\frac{\beta_2\, r}{\alpha_2}} (\alpha_1+\alpha_2\, r)^{\frac{\beta_1}{\alpha_2}-\frac{\alpha_1 \beta_2}{\alpha_2^2}-\frac{\beta_0}{\alpha_1}-1}}\dfrac{d}{dr}\left[r^{\frac{\beta_0}{\alpha_1}} e^{\frac{\beta_2\, r}{\alpha_2}} (\alpha_1+\alpha_2\, r)^{\frac{\beta_1}{\alpha_2}-\frac{\alpha_1 \beta_2}{\alpha_2^2}-\frac{\beta_0}{\alpha_1}}\dfrac{df(r)}{dr}\right]-\varepsilon_1\, r\, f(r)=\varepsilon_0\, f(r),
\end{align*}
 where    $\{\alpha_i, \beta_i, \varepsilon_i\}$ are real parameters for which $\alpha_2$ and $\alpha_1$ are not simultaneously zero with no common factor between the polynomial coefficients in \eqref{eq3}. The differential equation \eqref{eq3} has two regular singular points
$r_0=0$ and
$r_1=- \alpha_1/\alpha_2$  with the exponents at the singularity  $\{0,1-\beta_0/\alpha_1\}$ and $\{0,1+(\beta_0/\alpha_1) - \beta/\alpha_2 + (\alpha_1 \beta_2)/\alpha_2^2\}$ in addition to one irregular singular point at $r_\infty=\infty$. If $\alpha_1\cdot\alpha_2\geq 0$, the domain of definition is $(0,\infty)$ while if $\alpha_1\cdot\alpha_2<0$, the domain of definition reads $(0,-\alpha_2/\alpha_1)$ with the assumption that the solution $f(r)$ vanishes at the domain boundary. 
\vskip0.1true in
\noindent  The applications of equation \eqref{eq3} are not limited to spectral problems generated by the softcore Coulomb potential: indeed, differential equations of the form \eqref{eq3} have played a key role for many diverse problems of theoretical physics \cite{exton,saad2009,ron1995,ru2011,chris2011,cheb2004,jaick2008,Ma:Po,fiz2010,fiz2011,Rh:Ns:Ks2011,Rh:Ns:Ns,Hc:Rh}.  Hence, the problem studied here contributes to the ongoing analysis of Heun's equation and its confluent cases, bridging elusive problems  in physics to mathematical analysis \cite{ron1995}. 
\vskip0.1true in
\noindent In addition to an irregular singular point at $r = \infty$, the differential equation (\ref{eq2}) has an additional regular singular
point at $r = 0$ with exponents given by the roots of the indicial equation 
$
4s(s-1)-(k-1)(k-3)=0,
$ namely $s=(k-1)/2$ and $s=(3-k)/2$. For large $r$, the differential equation \eqref{eq2} behaves as 
\begin{equation*}
-\dfrac{d^2\psi_{{\mathfrak n}k}(r)}{dr^2}=E_{{\mathfrak n}k}(v) \,  \psi_{{\mathfrak n}k}(r).
\end{equation*}
\noindent The asymptotic solution allows us to factorize the wave function solution of \eqref{eq2} in the  form  
\begin{equation}\label{eq4}
\psi_{{\mathfrak n}k}(r)=r^{(k-1)/2}\exp\left(-\mathscr{E}_{{\mathfrak n}k}\,r\right)f(r),\qquad \mathscr{E}_{{\mathfrak n}k}=\sqrt{-E_{{\mathfrak n}k}(v)}>0.
\end{equation}
Substituting this ansatz solution into equation \eqref{eq2} gives for $f(r)$ the following differential equation 
\begin{align}\label{eq5}
(r^2+r)f''(r)&+\left(-2\mathscr{E}_{{\mathfrak n}k}\,r^2+(k-1-2\,\mathscr{E}_{{\mathfrak n}k})\,r+k-1\right)f'(r)-\left(((k-1)\,\mathscr{E}_{{\mathfrak n}k}- v)\,r+(k-1)\,\mathscr{E}_{{\mathfrak n}k}\right)f(r)=0,\end{align}
%or in self-adjoint form,
%\begin{align*}
%\dfrac{d}{dr}\left(e^{-2 \mathscr{E} r} r^{k-1}\dfrac{df}{dr}\right)-\frac{ r^{k-2} \,e^{-2\mathscr{E} r}}{1+r}(\mathscr{E} (1-k)+r (v+\mathscr{E} (1-k)))f(r)=0,
%\end{align*}
that is clearly a special case of the generalized confluent Heun equation \eqref{eq3}. 
\vskip0.1true in
\noindent The purpose of this article is threefold. First, we analyze the conditions on the equation parameters of \eqref{eq3} so as to admit the polynomial solutions $f_n(r)=\sum_{j=0}^n \mathcal C_j r^j$, $n=0,1,2,\dots$, and we set up a scheme to evaluate the polynomial coefficients $\mathcal C_k$ explicitly.  Second, we establish a correspondence between each solution $f_n(r)$ and a finite sequence of orthogonal polynomials $\{\mathpzc{P}_j(\varepsilon_0)\}_{j=0}^n$ for which we consider some mathematical properties, particularly the three-term recurrence relation, Christoffel-Darboux formulas, and the moments of the weight function. In addition to these properties, we also show that these polynomial solutions exhibit a factorization property similar to that of the Bender-Dunne Orthogonal Polynomials \cite{bender1996}. This factorization property, in turn, allows the construction of   associated (infinite) classes of orthogonal polynomials.  Thirdly, as a practical use of these polynomial solutions, we study the quasi-exact solvability of the $d$-dimensional Schr\"odinger equation with the scaled interaction $V(r)=-v/(r+1)$ as given by \eqref{eq2} which, for $b > 0$, is essentially equivalent to the unscaled form $V(r) = - e^2 Z/(r + b).$ 
\vskip0.1true in
\noindent The present article is organized as follows. In Section II, we explore the solvability constraints on the equation \eqref{eq3} and we design a procedure for constructing the polynomial solutions $f_n(r)=\sum_{j=0}^\infty \mathcal C_j \,r^k$ for each $n$. In Section III, we establish a correspondence between that the solution $y_n$ and a set of orthogonal polynomials $\{\mathpzc{P}_j\}_{j=0}^n$. The recurrence relations, the Christoffel-Darboux formulas, the weight functions, and the factorization property are presented.   In Section IV, we apply these results in the study of the $d$-dimensional radial Schr\"odinger equation with the softcore Coulomb potential given by equation \eqref{eq2}. 
%%%%%%%%%%%%%%%%%%%%%%%%%%%%%%%%%%%%%%%%%%
\section{Generalized confluent Heun equation: polynomial solutions}
%%%%%%%%%%%%%%%%%%%%%%%%%%%%%%%%%%%%%%%%%%
\noindent The $n^{th}$ derivative of \eqref{eq3} reads 
\begin{align}\label{eq6}
& \left[ a_{2} \,r^{2} + a_{1}\, r\, \right]\, f^{(n+2)} (r)\, + \, \left[ \beta_2\,r^2+ \left( \beta_1+2\,n\,a_2 \right) r+n\,a_1+\beta_0 \right]\, f^{(n+1)}(r) \notag\\
& \qquad + \big[ (2\,n\,\beta_2-\varepsilon_1)\,r+n\,(n-1)\,a_2+n\,\beta_1-\varepsilon_0 \big] \,f^{(n)}(r) + \big[ n\,(n-1)\,\beta_2 - n\,\varepsilon_1  \big] \,f^{(n-1)}(r)=0 \, ,\quad n=0,1,\dots \, .
\end{align}
Thus, for the existence of an $n^{th}$ degree polynomial solution $f_n(r)=\sum_{j=0}^n\mathcal{C}_j\,r^j$, it is necessary that 
\begin{align}\label{eq7}
\varepsilon_{1;n} \equiv \varepsilon_1=n\, \beta_2, \qquad (n=0,1,2,\dots).
\end{align}
On the other hand, an application of the Frobenius Method establishes for the coefficients of the polynomial solutions $f_n(r)=\sum_{k=0}^n{\mathcal C}_k\,r^k$ to 
the differential equation
\begin{align*}
\frac{\dfrac{d}{dr}\left[r^{\beta_0/\alpha_1} e^{\beta_2\, r/\alpha_2}(\alpha_1+\alpha_2\, r)^{-\frac{\alpha_1\beta_2}{\alpha_2^2}-\frac{\beta_0}{\alpha_1}+\frac{\beta_1}{\alpha_2}}\dfrac{df_n(r)}{dr}\right]}{r^{\frac{\beta_0}{\alpha_1}-1} e^{\frac{\beta_2 r}{\alpha_2}} (\alpha_1+\alpha_2 r)^{-\frac{\alpha_1\beta_2}{\alpha_2^2}-\frac{\beta_0}{\alpha_1}+\frac{\beta_1}{\alpha_2}-1}}
-n\,\beta_2\,r\, f_n(r)=\varepsilon_{0,n}\,f_n(r),\qquad n=0,1,2,\dots,
\end{align*}
the following three-term recurrence relation
\begin{equation}\label{eq8}
\big[ (j+1)\left(j\,a_1+\beta_0\right) \big] \,\mathcal{C}_{j+1} + \big[ j\,(j-1)a_2+j\,\beta_1-\varepsilon_{0,n} \big] \,\mathcal{C}_{j} +\big[ \left(j-1\right)\beta_{2}-\varepsilon_{1,n} \big] \,\mathcal{C}_{j-1} = 0 \, , \qquad j=0,1,2,\dots,n+1,
\end{equation} 
which allows the explicit computation of all the coefficients $\{\mathcal{C}_j\}_{j=0}^n$ of the polynomial solutions initiated with  $\mathcal{C}_{-1}=0$ and $\mathcal{C}_{0}=1$. 

\vskip0.1true in
\noindent  For each $n$, the recurrence relation \eqref{eq8} generates $n+2$ linear equations in the coefficients $\mathcal{C}_{j}, j=0,1,\dots, n$. The first $n$ equations may used to evaluate, using say Cramer's rule,  the coefficients $\mathcal{C}_{j}$, $j=1,2,\dots,n$, in terms of the non-zero coefficient $\mathcal{C}_{0}$. The equation indicated by $j=n$ gives the sufficient (polynomial) condition for $\varepsilon_{0,n}$ in terms of the other equation parameters which guarantees the existence of a polynomial solution. The very last equation of the linear system indicated by $j=n+1$ replicate the necessary condition \eqref{eq7}. It is not difficult to show that the sufficient (polynomial) condition is equivalent to the vanishing of the $(n+1)^{th}$-order determinant given by the Jacobi Tri-diagonal Matrix:\\

\begin{minipage}{0.5\textwidth} \centering
	\begin{center}
$\Delta_{n+1}$~~=~~\begin{tabular}{|lllllll|}
 $S_0$&$T_1$& & & & & \\
 $\gamma_1$&$S_1$&$T_2$& & & & \\
 & $\gamma_2$&$S_2$&$T_3$& & & \\
 & &$\ddots$&$\ddots$&$\ddots$& &\\ 
 & & &$\gamma_{n-2}$&$S_{n-2}$&$T_{n-1}$&  \\
 & & & &$\gamma_{n-1}$&$S_{n-1}$&$T_n$\\
 & & & & &$\gamma_{n}$&$S_n$\\
\end{tabular}~~, \quad
\end{center}
\end{minipage} \quad
\begin{minipage}{0.35\textwidth} for
\begin{align*}
S_j&=j \big((j-1)\alpha_2+\beta_1 \big)-\varepsilon_{0,n}\, ,\notag\\
T_j&=j \big((j-1)\alpha_1+\beta_0 \big)\, ,\notag\\
\gamma_j&= - \big(n-j+1 \big)\,\beta_2\, ,\hskip1true in
\end{align*}
\end{minipage}
\vskip0.1true in
\noindent where all the other entries being zero.   
\vskip0.1true in
\noindent The determinant $\Delta_{n+1}=0$ gives an $n$-degree polynomial of the discrete variable $\varepsilon_{0,n_i}, 1\leq i\leq n$. Here, the index $i$ introduced in reference to distinct roots of  $\Delta_{n+1}=0$ for each given $n$. The values of $\varepsilon_{0,n_i}$ may be regarded as the roots of the square matrix $L_{n+1}=\Delta_{n+1}+\varepsilon_{0,n}\, I$ where $I$ is the $(n+1)$-identity matrix.  For the real square matrix $L_{n+1}$, the characteristic roots are real and simple if the product $T_i\,\gamma_i>0$ for all $i=1,2,\dots n.$  Further, between any two 
characteristic roots of $L_{n+1}$ lies exactly one characteristic root of $L_n$. On the other hand, if $T_i\,\gamma_i<0,~ i=1,2,\dots,n,$ then all the real characteristic roots of $L_{n+1}$ lie between the least and the greatest of the main diagonal entries $\mathfrak s_j\equiv S_j+\varepsilon_{0,n}$. Further, if $\mathfrak s_0<\mathfrak s_1<\dots<\mathfrak s_{n+1}$, then the characteristic roots of $L_{n+1}$ and those of $L_{n}$ cannot interlace \cite{arscott,MM:HM,jay,gibson}. In short, the above discussion provide at our disposal a complete procedure to compute the $n^{th}$-degree polynomial solutions of the Generalized Confluent Heun Equation \eqref{eq3}.
\vskip0.1true in
\noindent Before closing this section, we report the first few polynomial solutions generated by the recurrence relation \eqref{eq8}: 
\vskip0.1true in
\noindent For the zero-degree polynomial solution $f_0(r)= 1$, $n=0,~j=0,1$ that gives  
the necessary and sufficient conditions, respectively: 
\begin{equation}\label{eq9}
\varepsilon_{1;0}=0,\qquad  \varepsilon_{0,0_1}=0,
\end{equation}
where the subindex $1$ in $\varepsilon_{0,0_1}$ refers to only one real root for $\Delta_{1}=0$.
\vskip0.1true in
\noindent For the first-degree polynomial solution $f_1(r)=\mathcal C_0+\mathcal C_1\,r$, $n=1$ and thus $j=0,1,2$. The coefficients $\mathcal C_j,j=0,1,2$ are evaluated using \eqref{eq8} that gives for $\mathcal C_0=1$ that $\mathcal C_1= \varepsilon_{0,1_i}/\beta_0$ and $\mathcal C_2= (\beta_0 \beta_2-\beta_1\varepsilon_{0,1}+\varepsilon_{0,1}^2)/(2 \beta_0(\alpha_1+\beta_0))$, whence
 \begin{equation}\label{eq10}
f_1(r)=1+\dfrac{\varepsilon_{0,1_i}}{\beta_0}\, r,\qquad i=1,2,
\end{equation}
subject to $\varepsilon_{1,1}=\beta_2$ and 
\begin{equation}\label{eq11}
\Delta_{2}=\left|\begin{array}{lll}
-\varepsilon_{0,1}&\beta_0\\
-\beta_2& \beta_1 -\varepsilon_{0,1}
\end{array}\right|\equiv \beta_0 \beta_2-\beta_1\varepsilon_{0,1_i}+\varepsilon_{0,1_i}^2=0,\qquad i=1,2.
\end{equation}
For the second-order polynomial solution $f_2(r)=\mathcal C_0+\mathcal C_1\,r+\mathcal C_2 \, r^2$, $n=2$ and thus $j=0,1,2,3$. The coefficients $\mathcal C_j,j=0,1,2,3$ are evaluated using \eqref{eq8} for $\mathcal C_0=1$ that $\mathcal C_1= \varepsilon_{0,2_i}/\beta_0$, 
$\mathcal C_2= (2\beta_0 \beta_2-\beta_1\varepsilon_{0,2_i}+\varepsilon_{0,2_i}^2)/(2 \beta_0(\alpha_1+\beta_0))$\\
and 
$$\mathcal C_3=\frac{{\varepsilon_{0,2_i}^3-(2 \alpha_2+3\beta_1)\varepsilon_{0,2_i}^2+2 (\beta_1 (\alpha_2+\beta_1)+\beta_2(\alpha_1+2 \beta_0)) \varepsilon_{0,2_i} -4(\alpha_2+\beta_1)\beta_0\beta_2}}{{6\beta_0(\alpha_1+\beta_0) (2\alpha_1+\beta_0)}},\qquad i=1,2,3$$ whence
 \begin{equation}\label{eq12}
f_2(r)=1+\dfrac{\varepsilon_{0,2_i}}{\beta_0}\, r +\frac{2 \beta_2\beta_0-\beta_1\varepsilon_{0,2_i}+\varepsilon_{0,2_i}^2}{2\beta_0(\beta_0+ a_1)}\,r^2,\qquad i=1,2,3,
\end{equation}
subject to $\varepsilon_1=2\beta_2$ and
\begin{align}\label{eq13}
\Delta_{3}&=\left|\begin{array}{lll}
-\varepsilon_{0,2_i}&\beta_0&0\\
-2\beta_2& \beta_1 -\varepsilon_{0,2_i}&2a_1+2\beta_0\\
0&-\beta_2&2a_2+2\beta_1-\varepsilon_{0,2_i}
\end{array}\right|\notag\\
&\equiv
\varepsilon_{0,2_i}^3 -(2 \alpha_2+3\beta_1)\varepsilon_{0,2_i}^2+2 (\beta_1 (\alpha_2+\beta_1)+\beta_2(\alpha_1+2 \beta_0)) \varepsilon_{0,2_i}-4 \beta_0\beta_2 (\alpha_2 + \beta_1) 
=0,\qquad i=1,2,3.
\end{align}
In summary, we have the following theorem..
\begin{theorem} \label{ThmII.1} (Polynomial solutions)
The necessary and sufficient conditions for $n$-degree polynomial solutions $f_n(r)=\sum_{j=0}^n C_j\,r^j$ of the differential equation \eqref{eq3}, respectively, are
\begin{equation}\label{eq14}
\varepsilon_{1;n}=n\, \beta_2,\qquad n=0,1,\dots,
\end{equation}
and the vanishing of the $\Delta_{n+1}$-determinant \\ \\
\begin{minipage}{0.5\textwidth} \centering
	\begin{center}
$\Delta_{n+1}$~~=~~\begin{tabular}{|lllllllll|}
 $S_0$ & $T_1$ &$~$&~& ~&~ &~& & \\
  $\gamma_1$ & $S_1$ &  $T_2$&$~$&~&~&~& & \\
~ & $\gamma_2$  & $S_2$&$T_3$&$~$&~&~& &\\
$~$&~&$\ddots$&$\ddots$&$\ddots$&~&~& &\\ 
~&~&~& & &$\gamma_{n-2}$&$S_{n-2}$&$T_{n-1}$&$~$ \\
~&~&~&&~& &$\gamma_{n-1}$&$S_{n-1}$&$T_n$\\
~&~&~&~&$~$& & &$\gamma_{n}$&$S_n$\\
\end{tabular}~~, \quad
\end{center}
\end{minipage} \quad
\begin{minipage}{0.47\textwidth}
\begin{align}
&for\notag\\
&S_j=j((j-1)\alpha_2+\beta_1)-\varepsilon_{0,2_i},\notag\\
& T_j=j((j-1)\alpha_1+\beta_0),\notag\\
 &\gamma_j=(j-1)\,\beta_2-\varepsilon_{1;2}.\label{eq15}
\end{align}
\end{minipage}
\end{theorem}
\noindent Noting that, a simple relation to evaluate this determinant $\Delta_{n+1}=0$, in terms of the lower-degree determinants, is
\begin{align}\label{eq16}
\Delta_{-1}&=0, ~\Delta_0=1,\notag\\ 
\Delta_{j+1}&=(j((j-1)\alpha_2+\beta_1)-\varepsilon_{0})\,\Delta_j-j\beta_2(n-j+1)\,((j-1)\alpha_1+\beta_0)\,\Delta_{j-1},
\end{align}
for $j=1,2,\dots,n.$  This Theorem can be illustrated by the following diagram:
\vskip0.1true in

\begin{tikzpicture}[edge from parent/.style={draw,-latex}]
\node {$\overset{n=0}{(\varepsilon_{1;0})}$}
child { node {$\varepsilon_{0;0_1}$} 
child  { node {$f_{0_1}(r)$} }};
\hskip1.0true in
\node {$\overset{n=1}{(\varepsilon_{1;1})}$}
child { node {$\varepsilon_{0;1_1}$} child  { node {${f}_{1_1}(r)$} }}
child { node {$\varepsilon_{0;1_2}$} child  { node {${f}_{1_2}(r)$} }
  };
\hskip1.5true in
\node {$\overset{n=2}{(\varepsilon_{1;2})}$}
child { node {$\varepsilon_{0;2_1}$} child  { node {$f_{2_1}(r)$} }}
child { node {$\varepsilon_{0;2_2}$}child  { node {$f_{2_2}(r)$} }}
child { node {$\varepsilon_{0;2_3}$} child  { node {$f_{2_3}(r)$} }};  
\hskip2true in
\node {$\overset{n=3}{(\varepsilon_{1;3})}$}
child { node {$\varepsilon_{0;3_1}$} child  { node {$f_{3_1}(r)$} }}
child { node {$\varepsilon_{0;3_2}$}child  { node {$f_{3_2}(r)$} }}
child { node {$\varepsilon_{0;3_3}$} child  { node {$f_{3_3}(r)$} }} 
child { node {$\varepsilon_{0;3_4}$}child  { node {$f_{3_4}(r)$} }};
\hskip1.3true in
\node {$\cdots$}
child { node {$\cdots$} child {node {$\cdots$}}};
  \end{tikzpicture}
\vskip0.1true in
%%%%%%%%%%%%%%%%%%%%%%%%%%%%%%%%%%%%%%
\section{Finite sequences of orthogonal polynomials: mathematical properties}
%%%%%%%%%%%%%%%%%%%%%%%%%%%%%%%%%%%%%%
\noindent In the light of Theorem \ref{ThmII.1}, the polynomial solutions of the differential equation \eqref{eq3} can be written as
\begin{align}\label{eq17}
f_n(r)&=\sum\limits_{j=0}^n\dfrac{\mathpzc{P}_{j}^n(\varepsilon_0)}{j!\, a_1^j\,\left(\frac{\beta_0}{a_1}\right)_j}\,r^j,\quad n=0,1,2,\dots,
\end{align}
where $(\alpha)_n$ refers to the Pochhammer Symbol $(\alpha)_n=\alpha(\alpha+1)\dots(\alpha-n+1)$ defined in terms of Gamma Function as $(\alpha)_n=\Gamma(\alpha+n)/\Gamma(\alpha)$ and has the property $(-n)_k=0$ for any positive integers $k\geq n+1$.
\vskip0.1true in
\noindent  Denote $\varepsilon_0=\zeta$, for fixed $n$, each polynomial solution $f_n(r)$ \emph{generates} a finite sequence of the coefficients $\mathpzc{P}_{j}^n(\zeta),~j=0,\dots,n$ that can be evaluated independently using the three-term recurrence relation 
\begin{equation}\label{eq18}
\mathpzc{P}_{j+1}^n(\zeta)=\big[ \zeta-j\,(j-1)\,a_2-j\,\beta_1 \big] \, \mathpzc{P}_j^n(\zeta)+ \big[ j\,(n+1-j) \beta_2 \big( (j-1)a_1+\beta_0 \big) \big] \,\mathpzc{P}_{j-1}^n(\zeta) \, , \quad \mathpzc{P}_{-1}^n(\zeta)=0 \, , \quad \mathpzc{P}_{0}^n(\zeta)=1.
\end{equation}
From the classical theory of orthogonal polynomials, particularly Favard Theorem  \cite{Favard,chihara,ismail}, one confirms that $\{\mathpzc{P}_{k}^n(\zeta)\}_{k=0}^n$ is a finite sequence of orthogonal polynomials. In the remaining part of this section, we explore some mathematical properties of these families of orthogonal polynomials. \vskip0.1true in
\begin{theorem}  \label{ThmIII.1} (Christoffel-Darboux formula) For $\zeta_1\neq \zeta_2$,
\begin{align}\label{eq19}
\sum_{j=0}^m \frac{\mathpzc{P}_{j}^n(\zeta_1)\mathpzc{P}_j^n(\zeta_2)}{j!\,(\beta_2 a_1)^j\left(\frac{\beta_0}{a_1}\right)_j(-n)_j}
&= \frac{\mathpzc{P}_{m+1}^n(\zeta_1)\mathpzc{P}_{m}^n(\zeta_2)-\mathpzc{P}_{m}^n(\zeta_1) \mathpzc{P}_{m+1}^n(\zeta_2)}{m!\,(\beta_2 a_1)^m\left(\frac{\beta_0}{a_1}\right)_m(-n)_m (\zeta_1-\zeta_2)} \, , \qquad m=1,\dots, n-1 \, ,
\end{align}
while for $\zeta_2\rightarrow \zeta_1=\zeta$
\begin{align}\label{eq20}
\sum_{j=0}^m\frac{\big(\mathpzc{P}_{j}^n(\zeta)\big)^2}{j!\,(\beta_2 a_1)^j\left(\frac{\beta_0}{a_1}\right)_j(-n)_j}
&=\frac{[\mathpzc{P}_{k+1}^n(\zeta)]'\mathpzc{P}_{k}^n(\zeta)-[\mathpzc{P}_{k}^n(\zeta)]'\mathpzc{P}_{k+1}^n(\zeta)}{k!\,(\beta_2 a_1)^k\left(\frac{\beta_0}{a_1}\right)_k(-n)_k},\qquad m=1,\dots,n-1 \, .
\end{align}
\end{theorem}
\begin{proof} The recurrence relations \eqref{eq18} for $\zeta_1$ and $\zeta_2$ read:
\begin{align*}
\mathpzc{P}_{k+1}^n(\zeta_1)&=\zeta_1\mathpzc{P}_k^n(\zeta_1)- \big[ k(k-1)a_2+k\beta_1 \big] \mathpzc{P}_k^n(\zeta_1)- \big[ k\,(k-n-1) \beta_2 \big( (k-1)a_1+\beta_0 \big) \big] \mathpzc{P}_{k-1}^n(\zeta_1)\, , \quad {(i)} \\
\mathpzc{P}_{k+1}^n(\zeta_2)&=\zeta_2\mathpzc{P}_k^n(\zeta_2)- \big[ k(k-1)a_2+k\beta_1 \big] \mathpzc{P}_k^n
(\zeta_2)- \big[ k\,(k-n-1) \beta_2 \big( (k-1)a_1+\beta_0 \big) \big] \mathpzc{P}_{k-1}^n(\zeta_2) \, , \quad { (ii)} 
\end{align*}
respectively. Multiplying (i) by $\mathpzc{P}_{k}^n(\zeta_2)$ and  (ii) by $\mathpzc{P}_{k}^n(\zeta_1)$ then subtracting the latter from the prior yields:
\begin{align*}
(\zeta_1-\zeta_2)\mathpzc{P}_{k}^n(\zeta_2)
\mathpzc{P}_k^n(\zeta_1)=Q_{k+1}(\zeta_1,\zeta_2)-\lambda_{k+1}Q_{k}(\zeta_1,\zeta_2),
\end{align*}
where 
\begin{align*}
\lambda_{k+1}&=k\,(k-n-1) \beta_2 \big( (k-1)a_1+\beta_0 \big) \, , \quad \text{and} \quad Q_{k+1}(\zeta,\zeta')=\mathpzc{P}_{k+1}^n(\zeta_1)\mathpzc{P}_{k}^n(\zeta_2)-\mathpzc{P}_{k}^n(\zeta_1)\mathpzc{P}_{k+1}^n(\zeta_2) \, .
\end{align*}
Recursively over $k$, we have
\begin{align*}
(\zeta_1-\zeta_2)\mathpzc{P}_{k}^n(\zeta_1)\mathpzc{P}_k^n(\zeta_2)&=Q_{k+1}(\zeta_1,\zeta_2)-\lambda_{k+1}
Q_{k}(\zeta_1,\zeta_2)\\
(\zeta_1-\zeta_2)\mathpzc{P}_{k-1}^n(\zeta_1)\mathpzc{P}_{k-1}^n(\zeta_2)&=Q_{k}(\zeta_1,\zeta_2)-\lambda_{k}
Q_{k-1}(\zeta_1,\zeta_2)\\
\dots&=\dots\\
(\zeta_1-\zeta_2)\mathpzc{P}_0^n(\zeta_1)\mathpzc{P}_0^n(\zeta_2)&=Q_1(\zeta_1,\zeta_2) \, .
\end{align*}
From which, it is straightforward to obtain
\begin{align*}
(\zeta_1-\zeta_2)&\bigg[\mathpzc{P}_{k}^n(\zeta_1)\mathpzc{P}_k^n(\zeta_2)+
\lambda_{k+1}\mathpzc{P}_{k-1}^n(\zeta_1)\mathpzc{P}_{k-1}^n(\zeta_2)+
\lambda_{k+1}\lambda_{k}\mathpzc{P}_{k-2}^n(\zeta_1)\mathpzc{P}_{k-2}^n(\zeta_2)+
\lambda_{k+1}\lambda_{k}\lambda_{k-1}\mathpzc{P}_{k-3}^n(\zeta_1)\mathpzc{P}_{k-3}^n(\zeta_2)
+\dots\notag\\
&+
\lambda_{k+1}\lambda_{k}\lambda_{k-1}\lambda_{k-2}\dots \lambda_2\mathpzc{P}_0^n(\zeta_1)\mathpzc{P}_0^n(\zeta_2)\bigg]=Q_{k+1}(\zeta_1,\zeta_2) \, .
\end{align*}
We then divide both sides of this equation by $(\zeta_1-\zeta_2)\lambda_{k+1}\lambda_{k}\lambda_{k-1}\lambda_{k-2}\dots \lambda_2$, and sum over $k$ to obtain
\begin{align*}
\sum_{j=0}^k\frac{\mathpzc{P}_{j}^n(\zeta_1)\mathpzc{P}_j^n(\zeta_2)}{\lambda_{j+1}\lambda_{j}\lambda_{j-1}\dots \lambda_2}
&=(\lambda_{k+1}\lambda_{k}\lambda_{k-1}\dots\lambda_2 )^{-1}\frac{\mathpzc{P}_{k+1}^n(\zeta_1)\mathpzc{P}_{k}^n(\zeta_2)-\mathpzc{P}_{k}^n(\zeta_1)
\mathpzc{P}_{k+1}^n(\zeta_2)}{\zeta_1-\zeta_2} \, .
\end{align*}
The \eqref{eq19} follows, by noting that 
\begin{equation*}
\lambda_{k+1}\lambda_{k}\lambda_{k-1}\dots\lambda_2=\prod_{j=2}^{k+1}\lambda_j=k!\,
(\beta_2 a_1)^k\left(\frac{\beta_0}{a_1}\right)_k(-n)_k \, , \qquad k=1,\dots, n \, ,
\end{equation*}
and \eqref{eq20} follows by taking the limit of both sides of \eqref{eq19} as $\zeta_2\rightarrow \zeta_1=\zeta$.
\end{proof}
\noindent Another interesting property of the finite sequences $\{\mathpzc{P}_{k}^n(\zeta)\}_{k=0}^n$ deals with 
the normalized weight function $\mathpzc{W}(\xi)$,
\begin{equation}\label{eq21}
\int_{\mathpzc{S}} \mathpzc{W}(\xi)d\xi=1 \, ,
\end{equation}
for which the orthogonality  
\begin{equation}\label{eq22}
\int_{\mathcal S}  \mathpzc{P}_{k}^n(\zeta)\,\mathpzc{P}_{k'}^n(\zeta)\,\mathpzc{W}(\zeta)\,d\zeta=\eta_k^2\,\delta_{kk'} \, , \qquad k,k'=0,1,\dots,n \, ,
\end{equation}
is defined. Here, $\mathpzc{S}$ is the support of the measure $\mathpzc{W}(\zeta)\,d\zeta$ on the real line and $\delta_{kk'}$ is the classical Kronecker Delta Function $\delta_{kk}=1$, while for $k\neq k'$, $\delta_{kk'}\neq 0$.  
\begin{theorem} \label{ThmIII.2} For the monic orthogonal polynomials $\{ \mathpzc{P}_{k}^n(\zeta) \}_{k=0}^n$
\begin{equation}\label{eq23}
\int_{\mathcal S}  \zeta^{j} \mathpzc{P}_{k}^n(\zeta)\,\mathpzc{W}(\zeta)\,d\zeta= \eta_k^2\,\delta_{kj} \, , \qquad for~all\quad j\leq k \, .
\end{equation}
\end{theorem}
\begin{proof}
By writing $\mathpzc{P}_{k'}^n(\zeta)=\sum_{j=0}^{k'} \alpha_{k';j}\zeta^j$, where $\alpha_{k';k'}=1$, equation \eqref{eq22} yields
\begin{equation}\label{eq24}
\sum_{j=0}^{k'}\alpha_{k';j} \int_{\mathcal S} \zeta^j \mathpzc{P}_{k}^n(\zeta)\,\mathpzc{W}(\zeta)\,d\zeta=\eta_k\eta_{k'}\delta_{kk'} \, , \qquad k,k'=0,1,\dots,n \, , \qquad a_{k';k'}=1 \, .
\end{equation}
For $k'=0,1,\dots,n$, \eqref{eq24} generates a linear system that yields
\begin{align}\label{eq25}
\begin{pmatrix}
\int_{\mathcal S}  \mathpzc{P}_{k}^n(\zeta)\mathpzc{W}(\zeta)d\zeta \hfill \\ 
\int_{\mathcal S}  \zeta\mathpzc{P}_{k}^n(\zeta)\mathpzc{W}(\zeta)d\zeta \hfill \\ 
\int_{\mathcal S}  \zeta^2\mathpzc{P}_{k}^n(\zeta)\mathpzc{W}(\zeta)d\zeta\\ 
\int_{\mathcal S}  \zeta^3\mathpzc{P}_{k}^n(\zeta)\mathpzc{W}(\zeta)d\zeta\\
\vdots\\
\int_{\mathpzc{S}}  \zeta^{k'}\mathpzc{P}_{k}^n(\zeta)\mathpzc{W}(\zeta)d\zeta
\end{pmatrix}=\eta_k\begin{pmatrix}
1&0&0&0&\dots& 0\\
\alpha_{1;0}^n& 1&0&0&\dots&0\\ 
\alpha_{2;0}^n&\alpha_{2;1}^n&1&0&\dots&0\\ 
\alpha_{3;0}^n&\alpha_{3;1}^n&\alpha_{3;2}^n&1&\dots&0\\ 
\alpha_{4;0}^n&\alpha_{4;1}^n&\alpha_{4,2}^n&\alpha_{4,3}^n&\dots&0\\ 
\vdots&\vdots&\vdots&\vdots&\dots&\vdots\\ 
\alpha_{k';0}^n&\alpha_{k';1}^n&\alpha_{k';2}^n&\alpha_{k';3}^n&\dots&1
\end{pmatrix}_{(k'+1) \times (k'+1)}^{-1} \hspace*{-0.6cm} \times \hspace*{0.4cm} \begin{pmatrix}
\eta_0\delta_{k0}\\ 
\eta_1\delta_{k1}\\ 
\eta_2\delta_{k2}\\ 
\eta_3\delta_{k3}\\ 
\vdots\\
\eta_{k'}\delta_{kk'}
\end{pmatrix}_{(k'+1)\times 1} \, .
\end{align}
The inverse of the lower triangular matrix is, again, lower triangular with ones on the main diagonal. This argument establishes our assertion \eqref{eq22}. 
\end{proof}
\vskip0.1true in
\begin{theorem}\label{ThmIII.3}
The norms of all polynomials
$\mathpzc{P}_{k}^n(\xi)$ with $k\geq n+1$ vanish.
\end{theorem}
\begin{proof}
\noindent The norms $\eta_k$ can be evaluated using the recurrence relations \eqref{eq18} after multiplying throughout by $\zeta^{k-1}\mathpzc{W}(\zeta)$ and integrating over $\zeta$. Then by Theorem \ref{ThmIII.2}, a simple two-term recursion relation for the squared norm is obtained
\begin{equation*}
\mathcal G_k=k\,(k-1-n)\beta_2 \big( (k-1)a_1+\beta_0 \big)\mathcal G_{k-1} \, ,
\end{equation*}
with a solution given by
\begin{equation}\label{eq26}
\mathcal G_k=|\eta_k|^2=\mathlarger{ \prod}_{j=1}^k \big[j\,(j-1-n)\beta_2 \big( (j-1)a_1+\beta_0 \big) \big]=k!\, (\beta_2a_1)^k (-n)_k\, \left(\frac{\beta_0}{a_1}\right)_k \, ,
\end{equation}
which is equal to zero for $k=n+1,n+2,\dots$, because of the Pochhammer identity $(-n)_k=0$ for all $k\ge n+1$.
\end{proof}
\noindent From Theorems \ref{ThmIII.2} and \ref{ThmIII.3}, it follows that
\begin{align}
&\int_{\mathcal S}  \mathpzc{P}_{k}^n(\zeta)\,\mathpzc{W}(\zeta)\,d\zeta= \delta_{k0}=0 \, , \quad for~all ~~k\geq 1 \, , \label{eq27}\\ 
&\int_{\mathcal S}  \mathpzc{P}_{j}^n(\zeta)\mathpzc{P}_{k}^n(\zeta)\,\mathpzc{W}(\zeta)\,d\zeta= 0 \, ,\quad (0\leq j<k\leq n) \, ,\label{eq28}\\
 &\int_{\mathcal S} \left[\mathpzc{P}_{k}^n(\zeta)\right]^2\,\mathpzc{W}(\zeta)\,d\zeta=k!\, (\beta_2a_1)^k (-n)_k\, \left(\frac{\beta_0}{a_1}\right)_k \, ,\quad 0\leq k\leq n \, ,\label{eq29}\\
&\int_{\mathcal S}  \zeta\, \mathpzc{P}_{k}^n(\zeta)\mathpzc{P}_{k-1}^n(\zeta)\,\mathpzc{W}(\zeta)\,d\zeta=k!\, (\beta_2a_1)^k (-n)_k\, \left(\frac{\beta_0}{a_1}\right)_k \, ,\qquad k=1,2,\dots \, ,\label{eq30}\\
&\int_{\mathcal S}  \zeta\, [\mathpzc{P}_{k}^n(\zeta)]^2\,\mathpzc{W}(\zeta)\,d\zeta=k \big( (k-1)a_2+\beta_1 \big)k!\, (\beta_2a_1)^k (-n)_k\, \left(\frac{\beta_0}{a_1}\right)_k \, ,\qquad k=0,1,2,\dots \, ,\label{eq31}\\
&\int_{\mathcal S}\zeta^2[\mathpzc{P}_{k}^n(\zeta)]^2W(\zeta)d\zeta=\bigg[
(k+1)\, \beta_2(k-n)\, \left(\beta_0+k\,a_1\right)+k\,(k-n-1) \beta_2 \big( (k-1)a_1+\beta_0 \big)\notag\\
&\hskip1.37true in+ \big( k(k-1)a_2+k\beta_1 \big)^2\, \bigg] k!\, (\beta_2a_1)^{k} (-n)_k\, \left(\frac{\beta_0}{a_1}\right)_k \, .\label{eq32}
\end{align}
\vskip0.1true in
\noindent Although we don't have a general formula for arbitrary $k$ that evaluates the moments 
\begin{align}\label{eq33}
\mu_k =\int_{\mathcal S}\zeta^k\, \mathpzc{W}(\zeta)\, d\zeta,\quad k=0,1,\dots,n \, ,
\end{align}
it is possible to compute $\mu_k$ recursively using Theorem \ref{ThmIII.2} and the recurrence relation \eqref{eq18}. It is not difficult to show that, for $\beta_2\,\beta_0<0$,
\begin{align}\label{eq34}
\int_{\mathcal S}\mathpzc{W}(\zeta)\,d\zeta=1 \, , \quad\int_{\mathcal S}\zeta \,\mathpzc{W}(\zeta)\,d\zeta=0 \, , \quad \int_{\mathcal S}\zeta^2\,\mathpzc{W}(\zeta)\,d\zeta=-n\,\beta_0\beta_2 \, , \quad \int_{\mathcal S} \zeta^3  \mathpzc{W}(\zeta)\,d\zeta=-n\,\beta_0\beta_1\beta_2 \, .
\end{align}
and for all $k\geq 2$,
\begin{equation}\label{eq35}
\int_{\mathcal S}\zeta\, \mathpzc{P}_k^n(\zeta)\, \mathpzc{W}(\zeta)\,d\zeta=0 \, .
\end{equation}
For example,
\begin{align}\label{eq36}
\int_{\mathcal S}\zeta^3\,  \mathpzc{W}(\zeta)\,d\zeta=\beta_1\eta_1^2 \, ,\qquad 
\int_{\mathcal S}\zeta^4\,  \mathpzc{W}(\zeta)\, d\zeta=n\,\beta_0\, \beta_2\, \big(2\, \alpha_1 \beta_2 (n-1)
 + \beta_0\, \beta_2\, (3n-2 )-\beta_1 \big) \, .
\end{align}
From this discussion, the recurrence relation \eqref{eq19} can be rewritten as
\begin{align}\label{eq37}
\mathpzc{P}_{-1}^n(\zeta)&=0 \, , \qquad \mathpzc{P}_{0}^n(\zeta)=1 \, , \qquad \mathpzc{P}_{1}^n(\zeta)=\zeta \, , \notag\\
\mathpzc{P}_{k+1}^n(\zeta)&=\left(\zeta-\dfrac{\int_{\mathcal S}\zeta[\mathpzc{P}_{k}^n(\zeta)]^2\,W(\zeta)\,d\zeta}{\int_{\mathcal S}[\mathpzc{P}_{k}^n(\zeta)]^2\,W(\zeta)\,d\zeta}\right)\mathpzc{P}_k^n(\zeta)-\left(\dfrac{\int_{\mathcal S}[\mathpzc{P}_{k}^n(\zeta)]^2\,W(\zeta)\,d\zeta}{\int_{\mathcal S}[\mathpzc{P}_{k-1}^n(\zeta)]^2\,W(\zeta)\,d\zeta}\right)\mathpzc{P}_{k-1}^n(\zeta)\, , \quad k=1,\dots, n-1 \, .
\end{align}
\vskip0.1true in
\noindent Another interesting property of the polynomials $\{\mathpzc{P}_k^n(\zeta)\}_{k=0}^n$, besides being an orthogonal sequence, is that when the parameter $n$ takes positive integer values the polynomials exhibit a
factorization property similar to the factorization of the Bender-Dunne orthogonal polynomials studied earlier \cite{bender1996, saad2006}. The factorization property occurs because the third-term
in the recursion relation \eqref{eq18} vanishes when $j=n+1$, so that all subsequent polynomials have a common factor, $\mathpzc{P}_{n+1}^n(\zeta)$, called a \emph{critical polynomial}. To illustrate this factorization property, consider the case of $n=1$:
\begin{align*}
\mathpzc{P}_1^1(\zeta)=\zeta \, ,\quad
\mathpzc{P}_2^1(\zeta)&=\zeta^2-\beta_1\zeta +\beta_2 \beta_0 \, ,\\  
\mathpzc{P}_3^1(\zeta)&=(\zeta-2 \beta_1-2a_2) \mathpzc{P}_2^1(\zeta) \, ,\\
\mathpzc{P}_4^1(\zeta)&=\left(\zeta^2-(5 \beta_1+8 a_2)\zeta+6 \beta_1^2-3\beta_2 \beta_0+18\beta_1 a_2+12 a_2^2-6 \beta_2a_1\right)\mathpzc{P}_2^1(\zeta) \, .
\end{align*}
Here, the critical polynomial $P_2^1(\zeta)$ is a factor of every other polynomial $P_{k+2}^n(\zeta), k=1,2,\dots$. For $n=2$: 
\begin{align*}
\mathpzc{P}_1^2(\zeta)=\zeta \, ,\quad 
\mathpzc{P}_2^2(\zeta)&=\zeta^2-\beta_1\zeta +2\beta_2 \beta_0 \, ,\\
\mathpzc{P}_3^2(\zeta)&=-4 \beta_2 \beta_1\beta_0-4\beta_2\beta_0a_2+\left(2\beta_1^2+4\beta_2 \beta_0+2 \beta_1a_2+2\beta_2a_1\right)\zeta-(3 \beta_1+2a_2) \zeta^2+\zeta^3 \, ,\\ 
\mathpzc{P}_4^2(\zeta)&=(\zeta-3\beta_1-6a_2)\mathpzc{P}_3^2(\zeta) \, ,\\ 
\mathpzc{P}_5^2(\zeta)&= \big( 12 \beta_1^2 - 4 \beta_2 \beta_0 + 60 \beta_1 a_2 + 72 a_2^2 - 
 12 \beta_2 a_1 - (7 \beta_1+18 a_2)\zeta + \zeta^2 \big) \mathpzc{P}_3^2(\zeta) \, .
\end{align*}
which shows that $P_3^n(\zeta)$,  the critical polynomial, factorizes each of the next sequence $P_{k+3}^n(\zeta), k=1,2,\dots$.  
\vskip0.1true in
\noindent Generally, the polynomials $\mathpzc{P}_{k+n+1}^n(\zeta)$ beyond some critical polynomial $\mathpzc{P}_{n+1}^n(\zeta)$ are factored into the product 
\begin{align}\label{eq38}
\mathpzc{P}_{k+n+1}^n(\zeta)=\mathcal Q_k^n(\zeta)\,\mathpzc{P}_{n+1}^n(\zeta),\qquad k=0,1,\dots. 
\end{align}
The quotient polynomials $\{\mathcal Q_k^n(\zeta)\}_{k\geq 0}$  are generated using the following three-term recurrence relation
\begin{equation}\label{eq39}
\mathcal Q_k^n(\zeta)=\left[\zeta- (k+n)\big( (k+n-1)a_2+\beta_1\big)\right]\mathcal Q_{k-1}^n(\zeta)-\,\left[ (k-1)(k+n) \beta_2 \big( (k+n-1)a_1+\beta_0 \big) \right] \mathcal Q_{k-2}^{n}(\zeta) \, , \quad k\geq 1 \, ,
\end{equation}
initiated with $\mathcal Q_{-1}^n(\zeta)=0$, and $\mathcal Q_{0}^n(\zeta)=1$.
This recurrence relation follows by substituting \eqref{eq38} into \eqref{eq18} after re-indexing the polynomials to eliminate the common factor $\mathpzc{P}_{n+1}^n(\zeta)$ from both sides. Hence, the quotient polynomials $\mathcal Q_k^n(\zeta)$ form an infinite sequence of orthogonal polynomials for each value of $n$. The Christoffel-Darboux Formula for this sequence of orthogonal polynomials reads
\begin{align}\label{eq40}
\sum_{j=0}^k\frac{{\mathcal Q}_j^n(\zeta_1){\mathcal Q}_j^n(\zeta_2)}{j!(n+2)_j(\beta_2a_1)^j\left(n+1+\frac{\beta_0}{a_1}\right)_j}=\frac{{\mathcal Q}_{k+1}^n(\zeta_1){\mathcal Q}_{k}^n(\zeta_2)-{\mathcal Q}_{k}^{n}(\zeta_1){\mathcal Q}_{k+1}^n(\zeta_2)}{k!(n+2)_k(\beta_2a_1)^k\left(n+1+\frac{\beta_0}{a_1}\right)_k(\zeta_1-\zeta_2)} \, ,
\end{align}
and in the limit $\zeta_1\to \zeta_2=\zeta$,
\begin{align}\label{eq41}
\sum_{j=0}^k\frac{\left({\mathcal Q}_j^n(\zeta)\right)^2}{j!(n+2)_j(\beta_2a_1)^j\left(n+1+\frac{\beta_0}{a_1}\right)_j}=\frac{[{\mathcal Q}_{k+1}^n(\zeta)]'{\mathcal Q}_{k}^n(\zeta)-[{\mathcal Q}_{k}^{n}(\zeta)]'{\mathcal Q}_{k+1}^n(\zeta)}{k!(n+2)_k(\beta_2a_1)^k\left(n+1+\frac{\beta_0}{a_1}\right)_k} \, .
\end{align}
\begin{theorem}
The norms of all the polynomials 
$\mathpzc{Q}_{k}^n(\xi)$ are given by
\begin{equation}\label{eq42}
\mathcal G_k^{\mathcal Q}=k!\,(n+2)_k\,(\beta_2a_1)^k\left(n+1+\frac{\beta_0}{a_1}\right)_k.
\end{equation}
\end{theorem}
\begin{proof} The proof follows by multiplying the recurrence relation  \eqref{eq39} by $\zeta^{k-2}W(\zeta)$ and integrating over $\zeta$. This procedure implies the two-term recurrence relation 
 \begin{equation}\label{eq43}
\mathcal G_{k}^{\mathcal Q}=k\,(k+n+1) \beta_2 \big( (k+n)a_1+\beta_0 \big) \mathcal G_{k-1}^{\mathcal Q} \, ,\qquad k\geq 1 \, .
\end{equation}
where $\mathcal G_{k}^{\mathcal Q}=\int_{\mathcal S} |\mathcal Q_{k}^n(\zeta)|^2 \mathpzc{W}(z)dz=\int_{\mathcal S}\zeta^{k}Q_{k}^n(\zeta)d\zeta$,
 with a solution given by \eqref{eq42}.
\end{proof}
%%%%%%%%%%%%%%%%%%%%%%%%%%%%%%%%%%%%%%%
\section{Spectra generated by the interaction potential $V(r)=-v/(r+1) $}
%%%%%%%%%%%%%%%%%%%%%%%%%%%%%%%%%%%%%%%
\noindent The constructive solutions of the Generalized Confluent Heun Equation discussed in the previous section allow to establish the exact solutions of the Schr\"{o}dinger Equation \eqref{eq2} subject to certain values of the potential parameter.
Denote $\mathscr{E}=\sqrt{-E_{nk}}$, the polynomial solutions of equation \eqref{eq5} demand that
\begin{align}\label{eq44}
\sqrt{-E_{nk}} =\dfrac{v}{2n+k-1} \quad or \quad v=(2n+k-1)\, \mathscr{E} \, , \quad n=0,1,\dots \, , \quad k=d+2\ell \, , \quad d>1 \, .
\end{align}
For fixed $n$ the differential equation \eqref{eq5}, namely
\begin{align*}
\left(r^2+r\right) f''(r)+\left(-2 \sqrt{-E_{nk}}\, r^2+\left(-2 \sqrt{-E_{nk}}+k-1\right)\,r+k-1\right) f'(r)+\sqrt{-E_{nk}}\left(2\,n\, \,r+1-k\right)f(r)=0
\end{align*}
has polynomial solutions $f_n(r)=\sum_{j=0}^n \mathcal{C}_j\,r^j$ with polynomial coefficients $\mathcal{C}_j\equiv  \mathcal{C}_j(\mathscr{E})$  evaluated using the interesting three-term recurrence relation
\begin{align}\label{eq45}
\mathcal{C}_{-1}&=0 \, , \qquad \mathcal{C}_{0}=1 \, , \notag\\
\mathcal{C}_{j+1}&= \left[ \frac{2j+k-1}{(j+1)\left(j+k-1\right)}\left(\mathscr{E}-\frac{j\,(j+k-2)}{2j+k-1}\right) \right] \mathcal{C}_{j} - \left[ \frac{2\,\left(n+1-j\right)  }{(j+1)\left(j+k-1\right)} \mathscr{E}\right]\,\mathcal{C}_{j-1} \, , \quad j=0,1,\dots, n-1 \, .
\end{align}
These polynomial solutions may instead expressed as $f_n(r)=\sum\limits_{j=0}^n \mathpzc{P}_{j}^n(\mathscr{E})\,r^j / \left[ j!\,\left(k-1\right)_j \right]$ with coefficients, $\{\mathpzc{P}_{j}^n(\mathscr{E})\}_{j=0}^n$, evaluated using
\begin{align}\label{eq46}
\mathpzc{P}_{-1}^n(\mathscr{E})&=0 \, , \qquad \mathpzc{P}_{0}^n(\mathscr{E})=1 \, , \notag\\
\mathpzc{P}_{j+1}^n(\mathscr{E})&+ \big[ j(j+k-2)-\mathscr{E}(k+2j-1) \big] \mathpzc{P}_j^n(\mathscr{E})+ \big[ 2\,\mathscr{E}\,j\,(n+1-j) (j+k-2) \big] \mathpzc{P}_{j-1}^n(\mathscr{E}) =0  \, , \quad j=1,\dots, n-1 \, .
\end{align}
The interesting recurrence relations such as \eqref{eq45} and \eqref{eq46} have been the focus of several important investigations from the mathematical analysis point of view. We refer the reader to the work of Leopold \cite{leopold} for a thorough discussion regarding the possible characterization of the roots of such recurrence relations. Via \cite[Theorem 2.2]{leopold}, the roots of the polynomials generated by \eqref{eq46}, and therefore the eigenvalues of the corresponding Schro\"{o}dinger Equation \eqref{eq2}, are all simple, positive real, and may be generated by the  vanishing of the determinant
\vskip0.1true in
\begin{minipage}{0.5\textwidth} \centering
	\begin{center}
$\Delta_{n+1}$~~=~~\begin{tabular}{|lllllll|}
 $S_0$&$T_1$& & & & & \\
 $\gamma_1$&$S_1$&$T_2$& & & & \\
 & $\gamma_2$&$S_2$&$T_3$& & & \\
 & &$\ddots$&$\ddots$&$\ddots$& &\\ 
 & & &$\gamma_{n-2}$&$S_{n-2}$&$T_{n-1}$&  \\
 & & & &$\gamma_{n-1}$&$S_{n-1}$&$T_n$\\
 & & & & &$\gamma_{n}$&$S_n$\\
\end{tabular}~~, \quad
\end{center}
\end{minipage} \quad
\begin{minipage}{0.47\textwidth}
where\begin{align}
S_j&=j (k+j-2)-\mathscr{E}(k+2 j-1) \, ,\notag\\
T_j&=j(j+k-2) \, ,\notag\\
\gamma_j&=2 (n-j+1)\mathscr{E} \, .\label{eq47}
\end{align}
\end{minipage}
\vskip0.1true in
\noindent Further, between every two consecutive roots of the polynomial $\Delta_{n+1}=0$ lies a root of the polynomial $\Delta_n=0$ for $n=1,2,\dots$.  The first few such polynomial solutions are display in Table I. 
\begin{sidewaystable}
    \centering
\begin{tabular}{||l | l | l||} 
 \hline
 $n$ & $\{\mathpzc{P}_{j}^n(\mathscr{E})\}_{j=0}^{n+1}$,\qquad $\mathscr{E}=v/(2n+k-1),\quad n=0,1,2,\dots.$ & $f_n(r):$ Critical Polynomial $\mathpzc{P}_{n+1}^n(\mathscr{E})=0$\\ [0.5ex] 
 \hline\hline
 1 & $\mathpzc{P}_{0}^1(\mathscr{E})=1$ & $f_1(r)=1+ \mathscr{E}\,r$\\ 
~ & $\mathpzc{P}_{1}^1(\mathscr{E})=\mathscr{E}(k-1)$ & ~\\ 
~ & $\mathpzc{P}_{2}^1(\mathscr{E})=(k^2-1)\mathscr{E}(\mathscr{E}-1)$ & $\mathpzc{P}_{2}^1(\mathscr{E})=0$\\ [0.5ex]\hline
2 & $\mathpzc{P}_{0}^2(\mathscr{E})=1$ & $f_2(r)=1+ \mathscr{E}\,r+\frac{\mathscr{E}(\mathscr{E}(k+1)-k-3)}{2k}\,r^2$\\ 
~ & $\mathpzc{P}_{1}^2(\mathscr{E})=\mathscr{E}(k-1)$ & ~\\ 
 ~& $\mathpzc{P}_{2}^2(\mathscr{E})=\mathscr{E} (k-1) (\mathscr{E} (1+k)-k-3)$&~\\ 
 ~ & $\mathpzc{P}_{3}^2(\mathscr{E})=\mathscr{E} (1-k) (3+k) \left(\mathscr{E}(\mathscr{E}-3) (1+k)+2 k\right)$ & $\mathpzc{P}_{3}^2(\mathscr{E})=0$\\ [0.5ex]\hline
3 & $\mathpzc{P}_{0}^3(\mathscr{E})=1$ & $f_3(r)=1+ \mathscr{E}\,r+\frac{\mathscr{E}(\mathscr{E}(k+1)-k-5)}{2k}\,r^2$\\ 
~ & $\mathpzc{P}_{1}^3(\mathscr{E})=\mathscr{E}(k-1)$ & $-\frac{\mathscr{E}(\mathscr{E}^2 (k+1) (k+3)-3 \mathscr{E} (k+1) (k+5)+2 k (k+5))}{6k(k+1)}\,r^3$\\ 
 ~& $\mathpzc{P}_{2}^3(\mathscr{E})=\mathscr{E} (k-1) (\mathscr{E} (1+k)-k-5)$&~\\ 
 ~ & $\mathpzc{P}_{3}^3(\mathscr{E})=\mathscr{E} (1-k)\left(\mathscr{E}^2(k+1)(k+3)-3\mathscr{E}(1+k)(5+k)+2 k(5+k)\right)$ & ~\\ 
 ~& $ \mathpzc{P}_{4}^3(\mathscr{E})=\mathscr{E}(k-1) (k+5) \left((\mathscr{E}^3 -6\mathscr{E}^2) (k+1) (k+3)\right.$&\\
  & $\left.+\mathscr{E} \left(11 k^2+34 k+15\right)-6 k (k+1)\right)$& $\mathpzc{P}_{4}^3(\mathscr{E})=0$\\ \hline
 4 & $\mathpzc{P}_{0}^4(\mathscr{E})=1$ & $f_4(r)=1+ \mathscr{E}\,r+\frac{\mathscr{E}(\mathscr{E}(k+1)-k-7)}{2k}\,r^2$\\ 
~ & $\mathpzc{P}_{1}^4(\mathscr{E})=\mathscr{E}(k-1)$ & $-\frac{\mathscr{E}\left(\mathscr{E}^2 (k+1) (k+3)-3 \mathscr{E} (k+1) (k+7)+2 k (k+7)\right)}{6 k (k+1)}\,r^3$\\ 
 ~& $\mathpzc{P}_{2}^4(\mathscr{E})=\mathscr{E} (k-1) (\mathscr{E} (1+k)-k-7)$&$+\frac{\mathscr{E}\left(\mathscr{E}^3 (k+1) (k+3) (k+5)-6 \mathscr{E}^2 (k+1) (k+3) (k+7)+\mathscr{E} (k+3) (k+7) (11 k+7)-6 k (k+1) (k+7)\right)}{24 k (k+1) (k+2)}\,r^4$\\ 
 ~ & $\mathpzc{P}_{3}^4(\mathscr{E})=\mathscr{E} (1-k) \left(\mathscr{E}^2 (k+1) (k+3)-3 \mathscr{E} (k+1) (k+7)+2 k (k+7)\right)$ & ~\\ 
 ~& $ \mathpzc{P}_{4}^4(\mathscr{E})=\mathscr{E} (k-1) \left(\mathscr{E}^3 (k+1) (k+3) (k+5)-6 \mathscr{E}^2 (k+1) (k+3) (k+7)\right.$&~\\
 ~& $\left. +\mathscr{E} (k+3) (k+7) (11 k+7)-6 k (k+1) (k+7)\right)$&\\
 ~ &$\mathpzc{P}_{5}^4(\mathscr{E})=\mathscr{E} (1-k) (k+7) (\mathscr{E}^4 (k+1) (k+3) (k+5)-10 \mathscr{E}^3 (k+1) (k+3) (k+5)$&\\
 ~&$+\mathscr{E}^2 (k (5 k (7 k+57)+637)+339)-2\mathscr{E} (k (k (25 k+156)+239)+84)+24 k (k+1) (k+2))$ &$\mathpzc{P}_{5}^4(\mathscr{E})=0$ \\
[1ex] 
 \hline\end{tabular}
 \caption{The polynomial solutions of the differential equation \eqref{eq5} and the associated first few sequences of orthogonal polynomials $\{\mathpzc{P}_{j}^n(\mathscr{E})\}_{j=0}^{n+1}$}
\end{sidewaystable}
\vskip0.1true in
\noindent An illustration of the factorization property for the polynomials which generate  the exact solutions of the Schr\"{o}dinger Equation \eqref{eq2}, we report the following few examples: For $n=1$,
\begin{align*} 
\mathpzc{P}_{0}^1(\mathscr{E})=1,\quad  
\mathpzc{P}_{1}^1(\mathscr{E})&=(k-1) \mathscr{E},\\ 
\mathpzc{P}_{2}^1(\mathscr{E})&=(k-1)(k+1)\mathscr{E} (\mathscr{E}-1),\qquad (The~critical ~ polynomial) \\
\mathpzc{P}_{3}^1(\mathscr{E})&=\mathpzc{P}_{2}^1(\mathscr{E})(\mathscr{E}(k+3)-2k),\\ 
\mathpzc{P}_{4}^1(\mathscr{E})&=\mathpzc{P}_{2}^1(\mathscr{E})(\mathscr{E}^2 (k+3) (k+5)-\mathscr{E} (k+3) (5 k+1)+6 k (k+1)),\cdots &=\cdots
\end{align*}
where
\begin{align*} 
\mathpzc{Q}_{0}^1(\mathscr{E})&=1,\\ 
\mathpzc{Q}_{1}^1(\mathscr{E})&=\mathscr{E}(k+3)-2k,\\ 
\mathpzc{Q}_{2}^1(\mathscr{E})&=\mathscr{E}^2 (k+3) (k+5)-\mathscr{E} (k+3) (5 k+1)+6 k (k+1),\notag\\
\mathpzc{Q}_{3}^1(\mathscr{E})&=\mathscr{E}^3 (k+3) (k+5) (k+7)-3 \mathscr{E}^2 (k+3) (k+5) (3 k+1)+2 \mathscr{E} (k+1) \left(13 k^2+47 k+12\right)-24 k (k+1) (k+2),\\\cdots &=\cdots
\end{align*}
Further, for $n=2$,
\begin{align*} 
\mathpzc{P}_{0}^2(\mathscr{E})=1,\quad 
\mathpzc{P}_{1}^2(\mathscr{E})&=(k-1) \mathscr{E},\\
\mathpzc{P}_{2}^2(\mathscr{E})&=(k-1) \mathscr{E}(\mathscr{E}(k+1)-k-3),\\ 
\mathpzc{P}_{3}^2(\mathscr{E})&=(k-1)(k+3)\mathscr{E} (\mathscr{E}^2(k+1)-3\mathscr{E}(k+1)+2k),\qquad (The~critical ~ polynomial) \\
\mathpzc{P}_{4}^2(\mathscr{E})&=\mathpzc{P}_{3}^2(\mathscr{E})(\mathscr{E}(k+5)-3(k+1))\\ 
\cdots &=\cdots
\end{align*}
For which, the factorized polynomials read
\begin{align*}
\mathpzc{Q}_{0}^2(\mathscr{E})=1,\quad
\mathpzc{Q}_{1}^2(\mathscr{E})&=\mathscr{E} (k+5)-3 (k+1),\\
\mathpzc{Q}_{2}^2(\mathscr{E})&=\mathscr{E}^2 (k+5) (k+7)-\mathscr{E}(k+5) (7 k+9)+12 (k+1) (k+2),\\
\mathpzc{Q}_{3}^2(\mathscr{E})&=\mathscr{E}^3 (k+5) (k+7) (k+9)-6 \mathscr{E}^2(k+5) (k+7) (2 k+3)+\mathscr{E} (47 k^3+409 k^2+993 k+711)\\
&-60 (k+1) (k+2) (k+3),\\
\cdots &=\cdots
\end{align*}
Generally, the polynomials $\{\mathpzc{Q}_{j}^n\}_{j=0}^\infty$ are generated using the recurrence relation (see equation \eqref{eq39}):
\begin{align}\label{eq48}
\mathpzc{Q}_{-1}^n(\mathscr{E})&=0 \, , \qquad \mathpzc{Q}_{0}^n(\mathscr{E})=1 \, , \notag\\
\mathpzc{Q}_{j+1}^n(\mathscr{E})&+((n+j+1) (n+k+j-1) - \mathscr{E} ( 2 (n+j+1)+ k-1 ))\mathpzc{Q}_{j}^n(\mathscr{E})-2j\mathscr{E} (n+j+1+j) (n+k+j-1)\mathpzc{Q}_{j-1}^n(\mathscr{E})=0.
\end{align}
\vskip0.1true in
\noindent  We close this section by giving explicitly the first few exact eigenstate solutions to the Schr\"{o}dinger Equation \eqref{eq2}, namely,
\begin{equation}\label{eq49}
\left[- \frac{d^2}{dr^2}+ \frac{(k-1)(k-3)}{4\,r^2} - \dfrac{(2n+k-1)\mathscr{E}}{r+1}+{\mathscr{E}}^2\right] \psi_{{\mathfrak n}k}(r)=0,
\quad
\psi_{{\mathfrak n}k}(r)=r^{(k-1)/2}e^{-\,\mathscr{E}\,r} \, \sum_{j=0}^n \mathcal{C}_j(\mathscr{E})\, r^j.\end{equation}
(here $\mathfrak n$ is the number of the zeros in the wave function in contrast to the degree $n$ of the polynomial solutions). 
\vskip0.1true in
\noindent For $n=1$: the eigenenergy $\mathscr{E} = 1$ (or $E_{0k} = -1$) and thus the ground-state solutions, for $d>1$, of the Schr\"{o}dinger Equation,
\begin{align}\label{eq50}
\left[-\frac{d^2}{dr^2}+\frac{(k-1)(k-3)}{4\,r^2}-\dfrac{k+1}{r+1}\right] \psi_{{\mathfrak n}k}(r)=-\psi_{{\mathfrak n}k}(r) \, ,
\end{align}
are given by
\begin{align}\label{eq51}
\psi_{0k}(r)&=r^{(k-1)/2}e^{-r}\left[1+r\right], \, \qquad (k=d+2\ell).
\end{align}

\begin{figure}[ht]
\label{fig:fig1}
\centering
\includegraphics[width=5cm,height=4cm]{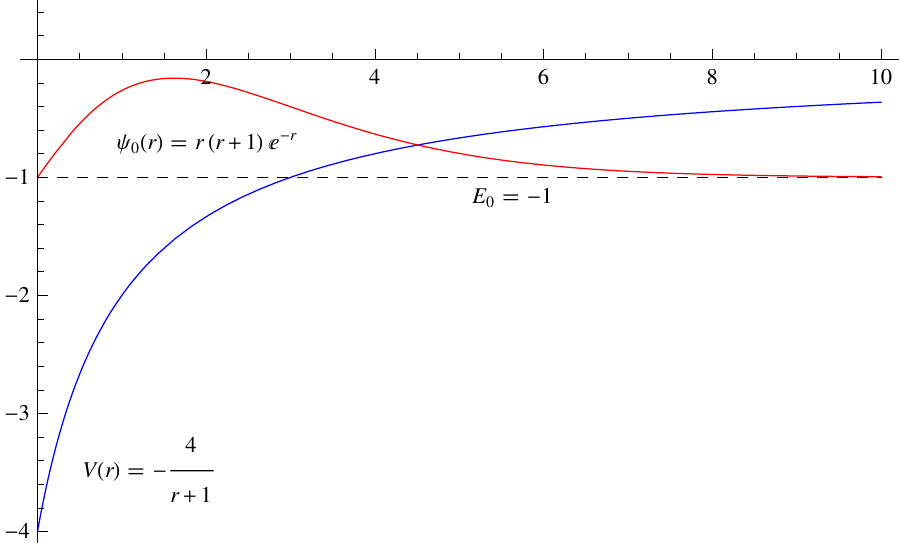} \qquad
\includegraphics[width=5cm,height=4cm]{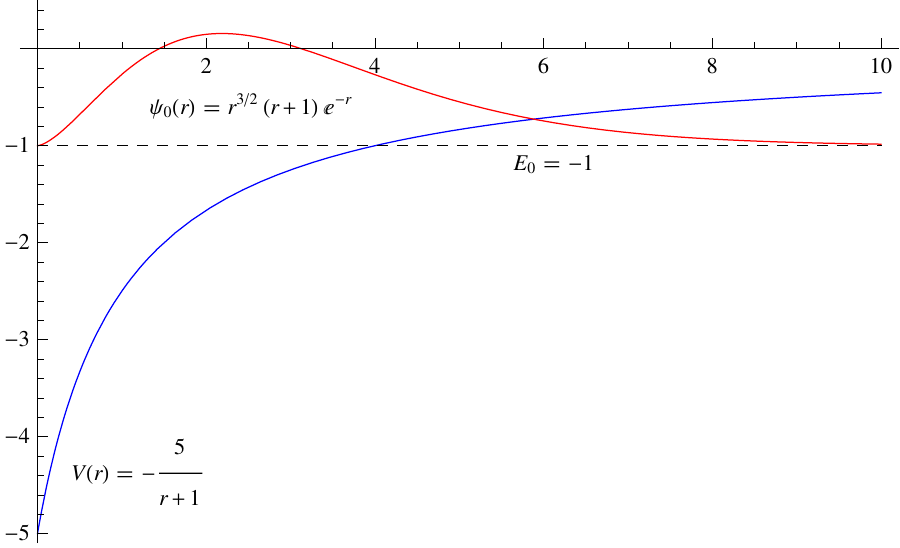} \qquad
\includegraphics[width=5cm,height=4cm]{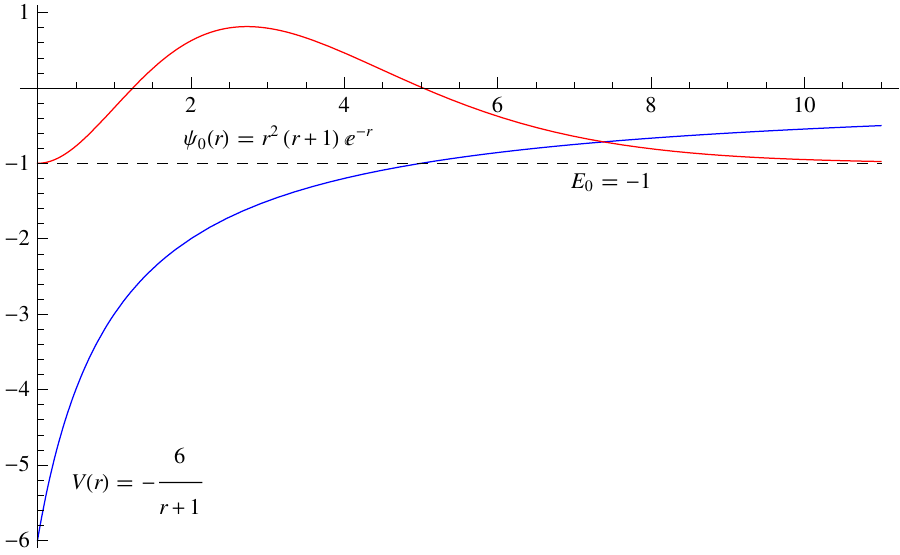} 
\caption{Plots of the potential $V(r) =-(k+1)/(1+r)$, $k=3,4,5$ along with the exact eigenvalue $E_0=-1$ and the exact (un-normalized)
wave function $\psi_0(r)$ for Schr\"odinger equation \eqref{eq50}.}
\end{figure}

\noindent{For $n=2$:} the eigenenergies $\mathscr{E}_i,i=1,2,$ of the Schr\"{o}dinger equation,
\begin{align}\label{eq52}
\left[- \frac{d^2}{dr^2}+ \frac{(k-1)(k-3)}{4\,r^2}-\dfrac{(k+3)\mathscr{E}_i}{r+1}\right] \psi_{{\mathfrak n}k}(r)=-\mathscr{E}_i^2\psi_{{\mathfrak n}k}(r) \, ,
\end{align}
are the roots of the quadratic equation,
\begin{align}\label{eq53}
(k+1)\, \mathscr{E}_i^2 -3\,(k+1)\mathscr{E}_i+2\, k=0 \, , \qquad i=1,2 \, ,
\end{align}
namely,
\begin{align}\label{eq54}
\mathscr{E}_1=\dfrac{3}{2}-\dfrac12\sqrt{\frac{k+9}{k+1}} \, , \qquad \text{and} \qquad \mathscr{E}_2=\dfrac{3}{2}+\dfrac12\sqrt{\frac{k+9}{k+1}}, \, \qquad (E_i\equiv E_{{\mathfrak n}k}=-\mathscr{E}_i^2)
\end{align}
with corresponding eigenstate solutions given by
\begin{align}\label{eq55}
\psi_{{\mathfrak n}k}(\mathscr{E}_i;r)&=r^{(k-1)/2}e^{-\,\mathscr{E}_i\,r}\left[1+\mathscr{E}_i r+\frac{\mathscr{E}_i (\mathscr{E}_i (1+k)-k-3)}{2\, k}\,r^2\right] \, , \quad i=1,2 \, .
\end{align}

\begin{figure}[ht]
\label{fig:fig2}
\centering
\includegraphics[width=7cm,height=5cm]{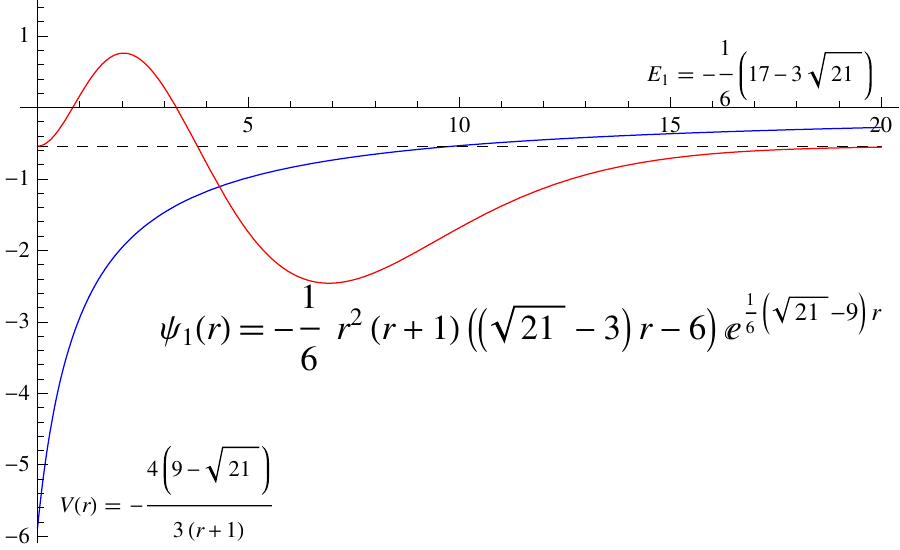} \qquad
\includegraphics[width=7cm,height=5cm]{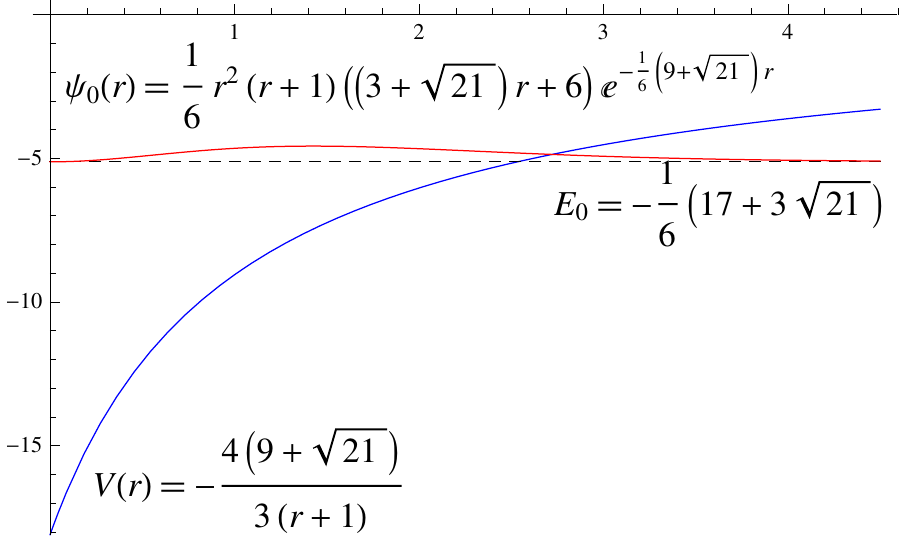}
\caption{Plot of the potential $V(r) =-{(k+3)\mathscr{E}_i}/{(r+1)}$, $k=5,i=1,2$ along with the exact eigenvalues $E_i,i=1,2$ and the exact (un-normalized)
wave functions \eqref{eq55} for Schr\"odinger equation \eqref{eq52}.}
\end{figure}

\noindent{For $n=3$:} the eigenenergies $\mathscr{E}_i,i=1,2,3,$ of the Schr\"{o}dinger Equation,
\begin{align}\label{eq56}
\left[- \frac{d^2}{dr^2}+ \frac{(k-1)(k-3)}{4\,r^2}-\dfrac{(k+5)\mathscr{E}_i}{r+1}\right] \psi_{{\mathfrak n}k}(r)=-\mathscr{E}_i^2\psi_{{\mathfrak n}k}(r) \, ,
\end{align}
are the roots of the cubic equation, 
\begin{align}\label{eq57}
\big[ (1+k) \,(3+k) \big] \mathscr{E}_i^3 - \big[ 6 \, (1+k) (3+k) \big] \mathscr{E}_i^2 + \big[ 15+34k+11 k^2 \big] \mathscr{E}_i -6 \,k \,(1+k)=0 \, , \qquad i=1,2,3 \, ,
\end{align}
with corresponding eigenstate solutions given by
\begin{align}\label{eq58}
\psi_{{\mathfrak n}k}(\mathscr{E}_i;r)&=r^{(k-1)/2}e^{-\,\mathscr{E}_i\,r}\notag\\
&\times\left[1+\mathscr{E}_i r+\frac{\mathscr{E}_i (\mathscr{E}_i (1+k)-k-5)}{2 k}\,r^2+\frac{\mathscr{E}_i  \left(  \mathscr{E}_i^2 (1+k) (3+k)
-3   \mathscr{E}_i(1+k) (5+k)+2 k (5+k)\right)}{6\,   k\, (1+k)}\,r^3\right] \, .
\end{align}
\noindent{For $n=4$:} the eigenenergies $\mathscr{E}_i,i=1,2,3,4,$ of the Schr\"{o}dinger Equation,
\begin{align}\label{eq59}
\left[- \frac{d^2}{dr^2}+ \frac{(k-1)(k-3)}{4\,r^2}-\dfrac{(k+7)\mathscr{E}_i}{r+1}\right] \psi_{{\mathfrak n}k}(r)=-\mathscr{E}_i^2\psi_{{\mathfrak n}k}(r)
\end{align}
are the roots of the quartic equation,
\begin{align}\label{eq60}
& \big[ (1+k) (3+k) (5+k) \big] \mathscr{E}_i^4 - \big[ 10 (1+k) (3+k) (5+k) \big] \mathscr{E}_i^3 \notag \\ 
& + \big[ 339 + 637 k + 285 k^2 + 35 k^3 \big] \mathscr{E}_i^2 -  \big[ 2 (84 + 239 k + 156 k^2 + 25 k^3) \big] \mathscr{E}_i + 24 k (1+k) (2+k)=0 \, , \quad i=1,2,3,4 \, ,
\end{align}
with the corresponding eigenstate solution given by
\begin{align}\label{eq61}
\psi_{{\mathfrak n}k}( \mathscr{E}_i;r)&
=r^{(k-1)/2}e^{-\,\mathscr{E}_i\,r}\notag\\
&\times\left[1+\mathscr{E}_i r+\frac{\mathscr{E}_i (\mathscr{E}_i (1+k)-k-7)}{2\, k}\,r^2+\frac{\mathscr{E}_i  \left(\mathscr{E}_i^2 (1+k) (3+k)
-3 \mathscr{E}_i (1+k) (7+k)+2 k (7+k)\right)}{6\,  k\, (1+k)}\,r^3\right.\notag\\
&\left.+\frac{\mathscr{E}_i^3 (1+k) (3+k) (5+k)-6  \mathscr{E}_i^2 (1+k) (3+k) (7+k)+  \mathscr{E}_i (3+k) (7+k) (7+11 k)-6 k (1+k) (7+k)}{24 k (1+k) (2+k)}\right]\,r^4 \, .
\end{align}
Higher order solutions may be obtained similarly using \eqref{eq45} or \eqref{eq46}.
%%%%%%%%%%%%%%%%%%%%%%%%%%%%%%%%%%%
\subsection{A note on the $3$-dimensional softcore Coulomb potential}
%%%%%%%%%%%%%%%%%%%%%%%%%%%%%%%%%%%
\noindent From the previous section, the eigenvalues of the Schr\"{o}dinger Equation,
\begin{equation}\label{eq62}
\left[-\dfrac{d^2}{dr^2}-\dfrac{v}{r+1}\right] \psi_{{\mathfrak n}k}(r)=-\mathscr{E}^2\,  \psi_{{\mathfrak n}k}(r) \, , \qquad 2\,(n+1)\,\mathscr{E}=v\qquad (v=b\, e^2 Z)\, ,
\end{equation}
are the roots of the polynomials 
\begin{align}\label{eq63}
n&=1\Longrightarrow \mathscr{E} -1=0 \, , \notag\\
n&=2\Longrightarrow 2\, \mathscr{E}^2 -6\mathscr{E}+3=0 \, , \notag\\
n&=3\Longrightarrow \mathscr{E}^3-6 \mathscr{E}^2+9\mathscr{E}-3=0 \, , \notag\\
n&=4\Longrightarrow 2 \mathscr{E}^4-20 \mathscr{E}^3+60 \mathscr{E}^2-60 \mathscr{E}+15=0 \, , \notag\\
n&= 5 \Longrightarrow 2\mathscr{E}^5-30 \mathscr{E}^4+150 \mathscr{E}^3-300 \mathscr{E}^2+225 \mathscr{E}-45=0 \, .
\end{align}
These polynomials are generated using the following three recurrence relations:
\begin{align} \label{eq64}
P_n(\mathscr{E})- \big[ 2(n+1)(\mathscr{E}-n) \big] P_{n-1}(\mathscr{E})+ \big[ n^2(n^2-1) \big] P_{n-2}(\mathscr{E})=0 \, , \quad n= 3,4,\dots \, ,
\end{align}
with initial conditions
\begin{equation*}
P_1(\mathscr{E})=8\mathscr{E}(\mathscr{E}-1) \, , \qquad \text{and} \qquad P_2(\mathscr{E})=24\mathscr{E}(3-6\mathscr{E}+2\mathscr{E}^2) \, .
\end{equation*}
Setting $P_n(\mathscr{E}) = (n+1)! Q_n(\mathscr{E})$ and dividing throughout by $(n+1)!$ yields the following recurrence relation
\begin{align} \label{eq65}
Q_n(\mathscr{E})- \big[ 2(\mathscr{E}-n) \big]\, Q_{n-1}(\mathscr{E}) + \big[ n(n-1) \big] Q_{n-2}(\mathscr{E})=0 \, , \quad n= 3,4,\dots \, , 
\end{align}
with initial conditions 
\begin{equation*}
Q_1(\mathscr{E})=4\mathscr{E}(\mathscr{E}-1) \, , \qquad \text{and} \qquad Q_2(\mathscr{E})=4\mathscr{E}(3-6\mathscr{E}+2\mathscr{E}^2) \, .
\end{equation*}
Finally, by writing $Q_n(\mathscr{E}) = 4\,E\, V_n(E)$, we see that $V_n(\mathscr{E})$ satisfies
\begin{align} \label{eq66}
V_n(\mathscr{E})- \big[ 2(\mathscr{E}-n) \big] V_{n-1}(\mathscr{E})+ \big[ n(n-1) \big] V_{n-2}(\mathscr{E})=0 \, , \quad n= 3,4,\dots \, ,
\end{align}
which is essentially the Laguerre polynomial 
\begin{align} \label{eq67}
V_n(\mathscr{E})=(-1)^n\,n!\,L_{n}^{(1)}(2\mathscr{E}) \, .
\end{align}
Thus, in the 3-dimenional case, the eigenvalues of the softcore Coulomb potential are given as the roots of the Laguerre Polynomials
\begin{align} \label{eq68}
L_{n}^{(1)}(2\mathscr{E})=0 \, .
\end{align}
The distribution of the the roots  for $P_j(\mathscr{E})\equiv L_j^{(1)}(2\mathscr{E})=0, j=1,2,\dots,6$ are illustrated in the following figure: 
\begin{figure}[!h]
\centering 
\includegraphics[width=\textwidth,height=1.5in]{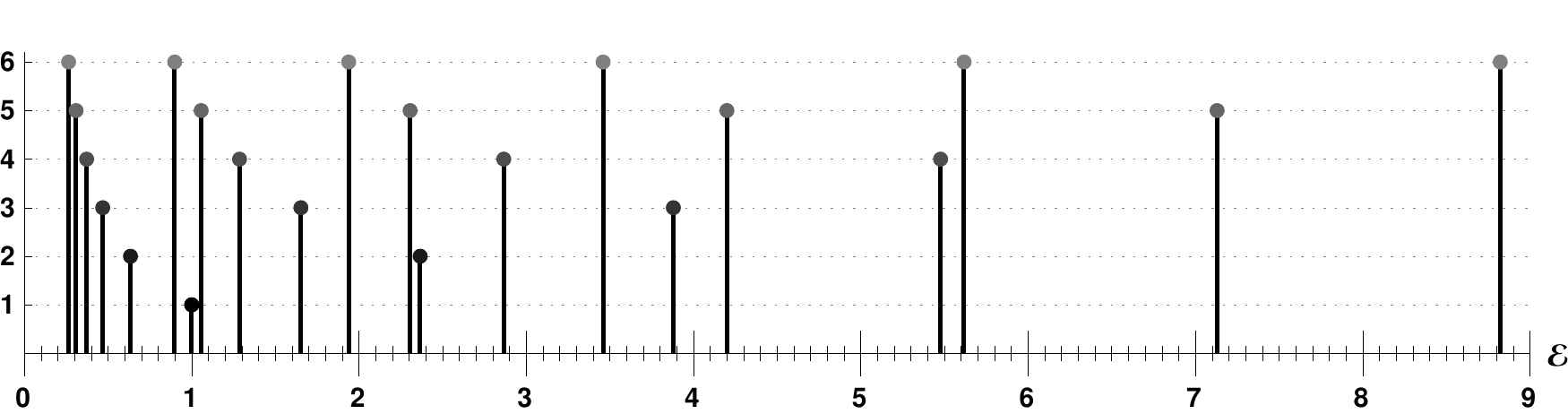}~~~~~~~~
\caption{The root distribution of the Laguerre Polynomials, $L_{n}^{(1)}(2\mathscr{E})=0$, for $n=1,2,3,4,5,6$.} 
\label{Fig:f1}
\end{figure}  

%------------------------------------------------------   
\section{Conclusions}
% ------------------------------------------------------
\noindent  The analysis of even apparently simple physical problems may lead to a wealth of mathematical structure and results. The softcore Coulomb potential studied in the present work is a case in point. The mathematics obtained by analyzing this elementary quantum system is remarkable: associated with each polynomial solution of the confluent Heun equation, there are two sequences of orthogonal polynomials, one finite and the other an infinite sequence generated by factoring out the finite one. The eigenenergies of the quantum system governed by Schr\"odinger's equation are the real distinct roots of the polynomial associated with the finite sequence. 

\section{Appendix}

\noindent For fixed $n$, the orthogonality of the polynomial solutions $f_{n_i}(r), 1\leq i\leq n$, of the Generalized Heun equation
\begin{align}\label{A.1}
\dfrac{1}{w(r)} \dfrac{d}{dr}\left[r^{\beta_0/\alpha_1} e^{\beta_2\, r/\alpha_2}(\alpha_1+\alpha_2\, r)^{-\frac{\alpha_1\beta_2}{\alpha_2^2}-\frac{\beta_0}{\alpha_1}+\frac{\beta_1}{\alpha_2}}\dfrac{df_{n_i}(r)}{dr}\right]
-n\,\beta_2\,r\, f_{n_i}(r)=\varepsilon_{0,n_i}\,f_{n_i}(r),\qquad n=0,1,2,\dots,
\end{align}
with respect to the weight function
\begin{align}\label{A.2}
w(r)=r^{\frac{\beta_0}{\alpha_1}-1} e^{\frac{\beta_2 r}{\alpha_2}} (\alpha_1+\alpha_2 r)^{\frac{\beta_1}{\alpha_2}-\frac{\alpha_1\beta_2}{\alpha_2^2}-\frac{\beta_0}{\alpha_1}-1},\qquad \left(\frac{\beta_0}{\alpha_1}>0,~ \frac{\beta_2}{\alpha_2}<0,~~\frac{\beta_1}{\alpha_2}-\frac{\alpha_1\beta_2}{\alpha_2^2}-\frac{\beta_0}{\alpha_1}\geq 0\right),
\end{align}
is given by
\begin{align}\label{A.3}
\left\{ \begin{array}{ll}
 \int_0^\infty f_{n_i}(r) f_{n_j}(r) w(r)dr=\delta_{ij}, &\mbox{ if $\alpha_1\cdot\alpha_2\geq 0$,} \\ \\
 \int_0^{-\alpha_1/\alpha_2} f_{n_i}(r) f_{n_j}(r) w(r)dr=\delta_{ij}, &\mbox{ if $\alpha_1\cdot \alpha_2<0$,}
       \end{array} \right.\qquad 
\end{align}
up to the normalization constant.  In reference to the domain definition, the normalization of the weight function 
$$\int_0^\infty w(r)dr=1, \qquad or\qquad \int_0^{-\alpha_1/\alpha_2} w(r)dr=1
$$ follows by means of the integral representation of the Kummer function \cite[Eq.13.4.4]{DLMF} to yields
\begin{align}\label{A.4}
\int_0^\infty r^{\frac{\beta_0}{\alpha_1}-1} e^{\frac{\beta_2}{\alpha_2}\,r} (\alpha_1+\alpha_2 r)^{\frac{\beta_1}{\alpha_2}-\frac{\alpha_1\beta_2}{\alpha_2^2}-\frac{\beta_0}{\alpha_1}-1}dr&=\alpha_1^{\frac{\beta_1}{\alpha_2}-\frac{\alpha_1\beta_2}{\alpha_2^2}-\frac{\beta_0}{\alpha_1}-1}\left(\frac{\alpha_1}{\alpha_2}\right)^{\beta_0/\alpha_1}  \Gamma \left(\frac{\beta_0}{\alpha_1}\right) U\left(\frac{\beta_0}{\alpha_1},\frac{\alpha_2\beta_1-\alpha_1\beta_2}{\alpha_2^2},-\frac{\alpha_1\beta_2}{\alpha_2^2}\right),\notag\\
&\hskip1.5true in (\alpha_1\beta_2/\alpha_2^2<0,~\beta_0/\alpha_1>0),
\end{align}
or
\begin{align}\label{A.5}
\int_0^{-\alpha_1/\alpha_2} r^{\frac{\beta_0}{\alpha_1}-1} &e^{\frac{\beta_2}{\alpha_2}\,r} (\alpha_1+\alpha_2 r)^{\frac{\beta_1}{\alpha_2}-\frac{\alpha_1\beta_2}{\alpha_2^2}-\frac{\beta_0}{\alpha_1}-1}dr\notag\\
&=\alpha_1^{\frac{\beta_1}{\alpha_2}-\frac{\beta_0}{\alpha_1}-\frac{\alpha_1\beta_2}{\alpha_2^2}-1} \left(-\frac{\alpha_1}{\alpha_2}\right)^{\frac{\beta_0}{\alpha_1}}
\frac{\Gamma \left(\frac{\beta_0}{\alpha_1}\right)\Gamma \left(-\frac{\beta_0}{\alpha_1}+\frac{\beta_1}{\alpha_2}-\frac{\alpha_1\beta_2}{\alpha_2^2}\right)}{\Gamma \left(\frac{\alpha_2\beta_1-\alpha_1\beta_2}{\alpha_2^2}\right)} {}_1F_1\left(\frac{\beta_0}{\alpha_1};\frac{\alpha_2\beta_1-\alpha_1\beta_2}{\alpha_2^2};-\frac{\alpha_1\beta_2}{\alpha_2^2}\right),\notag\\
&\hskip3true in(\alpha_1/\alpha_2<0,~\alpha_2^2 \beta_0+\alpha_1^2 \beta_2<\alpha_1\alpha_2\beta_1).
\end{align}
The normalization \eqref{A.3} of the exact solutions for the Generalized Heun equatiom
are evaluated using  either of the the integrals
\begin{align}\label{A.6}
\int_0^{-\alpha_1/\alpha_2}& r^{\frac{\beta_0}{\alpha_1}+i+j-1} e^{\frac{\beta_2}{\alpha_2}\,r} (\alpha_1+\alpha_2 r)^{\frac{\beta_1}{\alpha_2}-\frac{\alpha_1\beta_2}{\alpha_2^2}-\frac{\beta_0}{\alpha_1}-1}dr\notag\\
&= \alpha_1^{\frac{\beta_1}{\alpha_2}-\frac{\alpha_1\beta_2}{\alpha_2^2}-\frac{\beta_0}{\alpha_1}-1} \left(-\frac{\alpha_1}{\alpha_2}\right)^{\frac{\beta_0}{\alpha_1}+i+j} \frac{ \Gamma \left(\frac{\beta_0}{\alpha_1}+i+j\right) \Gamma \left(\frac{\alpha_1}{\alpha_2}-\frac{\beta_0}{\alpha_1}-\frac{\alpha_1\beta_2}{\alpha_2^2}\right)}{\Gamma \left(\frac{\alpha_2 \beta_1-\alpha_1 \beta_2}{\alpha_2^2}+i+j\right)} \, _1F_1\left(\frac{\beta_0}{\alpha_1}+i+j;\frac{\alpha_2 \beta_1-\alpha_1\beta_2}{\alpha_2^2}+i+j;-\frac{\alpha_1 \beta_2}{\alpha_2^2}\right),\notag\\
&\hskip3.5true in \left(\frac{\beta_0}{\alpha_1}+i+j>0.~ \frac{\beta_0}{\alpha_1}-\frac{\beta_1}{\alpha_2}+\frac{\alpha_1\beta_2}{\alpha_2^2}<0\right),
\end{align}
or
\begin{align}
\int_0^\infty &r^{\frac{\beta_0}{\alpha_1}+i+j-1} e^{\frac{\beta_2}{\alpha_2}\,r} (\alpha_1+\alpha_2 r)^{\frac{\beta_1}{\alpha_2}-\frac{\alpha_1\beta_2}{\alpha_2^2}-\frac{\beta_0}{\alpha_1}-1}dr\notag\\
&=\frac{\left(\frac{\alpha_1}{\alpha_2}\right)^{\frac{\beta_0}{\alpha_1}+i+j} \Gamma \left(\frac{\alpha_2 \beta_1-\alpha_1 \beta_2}{\alpha_2^2}+i+j-1\right)}{\alpha_1^{\frac{\alpha_1 \beta_2}{\alpha_2^2}+\frac{\beta_0}{\alpha_1}-\frac{\beta_1}{\alpha_2}+1}\left(-\frac{\alpha_1 \beta_2}{\alpha_2^2}\right)^{-\frac{\alpha_1 \beta_2-\alpha_2 (\alpha_2 (i+j-1)+\beta_1)}{\alpha_2^2}}}  \, _1F_1\left(\frac{\alpha_1 \beta_2}{\alpha_2^2}+\frac{\beta_0}{\alpha_1}-\frac{\beta_1}{\alpha_2}+1;\frac{\alpha_1\beta_2}{\alpha_2^2}-\frac{\beta_1}{\alpha_2}-i-j+2;-\frac{\alpha_1\beta_2}{\alpha_2^2}\right)\notag\\
&+\frac{ \left(\frac{\alpha_1}{\alpha_2}\right)^{\frac{\beta_0}{\alpha_1}+i+j} \Gamma \left(\frac{\beta_0}{\alpha_1}+i+j\right) \Gamma \left(-\frac{\beta_1}{\alpha_2}+\frac{\alpha_1 \beta_2}{\alpha_2^2}-i-j+1\right)}{\alpha_1^{\frac{\alpha_1 \beta_2}{\alpha_2^2}+\frac{\beta_0}{\alpha_1}-\frac{\beta_1}{\alpha_2}+1}\Gamma \left(\frac{\beta_0}{\alpha_1}+\frac{\alpha_1 \beta_2}{\alpha_2^2}-\frac{\beta_1}{\alpha_2}+1\right)} \, _1F_1\left(\frac{\beta_0}{\alpha_1}+i+j;-\frac{\alpha_1 \beta_2}{\alpha_2^2}+\frac{\beta_1}{\alpha_2}+i+j;-\frac{\alpha_1 \beta_2}{\alpha_2^2}\right),\notag\\
&\hskip4true in\left(\frac{\beta_0}{\alpha_1}+i+j>0, \frac{\alpha_1 \beta_2}{\alpha_2^2}<0\right),
\end{align}
according to the domain definition.

% ------------------------------------------------------   
\section{Acknowledgments}
% ------------------------------------------------------
\medskip
\noindent Partial financial support of this work under Grant Nos.  GP3438 and GP249507 from the Natural Sciences and
Engineering Research Council of Canada is gratefully acknowledged by us (respectively RLH and NS).
% --------------------------------------------------------------------------------

\section*{References}
\begin{thebibliography}{00}

%1
\bibitem{exton} H. Exton, \emph{The interaction $V(r) = -Ze^2/(r+\beta)$ and the confluent
Heun equation}, J. Phys. A: Math. Gen. \textbf{24} (1991) L329 - L330.
%2
\bibitem{saad2009} R. L. Hall, N. Saad, and K. D. Sen, \emph{Soft-core Coulomb potentials and Heun's differential 
equation}, J. Math. Phys. J. Math. Phys. \textbf{51} (2010) 022107. 

\bibitem{ron1995} A. Ronveaux (Ed.) \emph{Heun's Differential Equations}, The Clarendon Press Oxford University Press, New York (1995). 
%2
\bibitem{ru2011} R. Boyack and J. Lekner, \emph{Confluent Heun functions and separation of variables in spheroidal coordinates} J. Math. Phys. \textbf{52} (2011) 073517.
%5
\bibitem{chris2011} M. S. Cunha and H. R. Christiansen, \emph{Confluent Heun functions in gauge theories on thick braneworlds}, Phys. Rev. D \textbf{84} (2011) 085002.
%3
\bibitem{cheb2004}  E.S. Cheb-Terrab, \emph{Solutions for the General, Confluent and Biconfluent Heun equations and their connection with Abel equations}, J. Phys. A: Math. Gen. \textbf{37} (2004) 9923.
%9
\bibitem{jaick2008} L.J. El-Jaick, D. B. Bartolomeu D. B. Figueiredo, \emph{On certain solutions for Confluent and Double-Confluent Heun equations}, J. Math. Phys. \textbf{49} (2008) 083508.
%12
\bibitem{Ma:Po} G. Marcilhacy and R. Pons, \emph{The Schr\"{o}dinger equation for the interaction potential $x^2 + \lambda x^2/(1 + gx^2)$ and the first Heun confluent equation}, J. Phys. A: Math. Gen. \textbf{18} (1985) 2441-2449.
%10
\bibitem{fiz2010} P. P. Fiziev, \emph{Novel relations and new properties of confluent Heun's functions and their derivatives of arbitrary order}, J. Phys. A: Math. Theor. \textbf{43} (2010) 035203.
%11
\bibitem{fiz2011} P. P. Fiziev and D. Staicova, \emph{Application of the confluent Heun functions for finding the quasinormal modes of nonrotating black holes}, Phys. Rev. D \textbf{84} (2011) 127502.
%4
\bibitem{Rh:Ns:Ks2011} R. L. Hall, N. Saad and K. D. Sen, \emph{Discrete Spectra for Confined and Unconfined $-a/r + br^2$ potentials in $d$-Dimensions},
J. Math. Phys. \textbf{52} (2011) 092103. 

\bibitem{Rh:Ns:Ns} Richard L. Hall, Nasser Saad and K. D. Sen, \emph{Spectral characteristics for a spherically confined $-a/r + br^2$ potential}, 
J. Phys. A: Math. Theor. \textbf{44} (2011) 185307.


\bibitem{Hc:Rh} H. Ciftci, R. L. Hall, N. Saad, E. Dogu, \emph{Physical applications of second-order linear differential equations that admit polynomial solutions},
J. Phys. A: Math. Theor. 43 (2010) 415206. 
%1
\bibitem{bender1996} C. M. Bender, G. V. Dunne, \emph{Quasi-exactly solvable systems and orthogonal polynomials},  J. Math. Phys. \textbf{37} (1996) 6 - 11. 

\bibitem{Favard} J. Favard, \emph{Sur les polynomes de Tchebicheff,} C. R. Acad. Sci. Paris \textbf{200} (1935) 2052 - 2053.

\bibitem{chihara} T. S. Chihara, \emph{
An introduction to orthogonal polynomials}, Gordon and Breach, 1978;
reprinted, Dover, 2011.
%13
\bibitem{ismail} M. E. H. Ismail,
\emph{Classical and quantum orthogonal polynomials in one variable}, Cambridge
University Press, 2005; corrected reprint, 2009.

%8

\bibitem{leopold} E. Leopold, \emph{Location of the zeros of polynomials satisfying three-term recurrence
relations. III. Positive coefficients case,} J. Approx. Theory \textbf{43} (1985) 15-24.

\bibitem{arscott} F. M.  Arscott, \emph{Latent  roots of tridiagonal matrices,} Proc.  Edinburgh Math.  Soc. \textbf{12} (1961) 5 - 7. 

\bibitem{MM:HM} M. Marcus and H. Minc, \emph{A survey of matrix theory and matrix inequalities,} Dover Publications (2010), The characteristic roots of Jacobi matrices are discussed on p. 166.

\bibitem{jay} J. W. Jayne, \emph{An exclusion theorem for tridiagonal matrices,} Proc. Edinburgh Math. Soc. \textbf{16} (1969) 251 - 253.


\bibitem{gibson} P. M. Gibson, \emph{Eigenvalues of complex tridiagonal matrices,} Proc. Edinburgh Math. Soc., \textbf{17} (1971), 317 - 319.

\bibitem{saad2006} N. Saad, R. L. Hall and H. Ciftci, \emph{Sextic anharmonic oscillators and orthogonal polynomials,}
J. Phys. A: Math. Gen.
\textbf{39} (2006) 8477 - 8486.

\bibitem{DLMF} \emph{NIST Digital Library of Mathematical Functions}. http://dlmf.nist.gov/15.8, Release 1.0.17 of 2017-12-22. F. W. J. Olver, A. B. Olde Daalhuis, D. W. Lozier, B. I. Schneider, R. F. Boisvert, C. W. Clark, B. R. Miller, and B. V. Saunders, eds.


\end{thebibliography}
\end{document}